\documentclass[aps,pra,superscriptaddress,notitlepage, amsfonts,longbibliography, twocolumn]{revtex4-2}

\usepackage{qcircuit}
\usepackage[dvipdfmx]{graphicx}
\usepackage{amsmath,amssymb,amsthm,mathrsfs,amsfonts,dsfont}
\usepackage{subfigure, epsfig}
\usepackage{braket}
\usepackage{bm}
\usepackage{enumerate}
\usepackage{physics}
\newtheorem{thm}{Theorem}
\newtheorem{lem}{Lemma}

\usepackage[colorlinks,linkcolor=blue,citecolor=blue]{hyperref}

\newcommand{\red}[1]{\textcolor{black}{#1}}
\newcommand{\redd}[1]{\textcolor{black}{#1}}
\newcommand{\black}[1]{\textcolor{black}{#1}}

\makeatletter
\renewcommand*{\numberline}[1]{\hb@xt@1em{#1\hfil}} 
\makeatother

\graphicspath{{./fig/}}

\begin{document}

\title{Universal cost bound of quantum error mitigation based on quantum estimation theory}

\author{Kento Tsubouchi}
\email{tsubouchi@noneq.t.u-tokyo.ac.jp}
\affiliation{Department of Applied Physics, University of Tokyo, 7-3-1 Hongo, Bunkyo-ku, Tokyo 113-8656, Japan}

\author{Takahiro Sagawa}
\affiliation{Department of Applied Physics, University of Tokyo, 7-3-1 Hongo, Bunkyo-ku, Tokyo 113-8656, Japan}
\affiliation{Quantum-Phase Electronics Center (QPEC), The University of Tokyo, Tokyo 113-8656, Japan}

\author{Nobuyuki Yoshioka}
\email{nyoshioka@ap.t.u-tokyo.ac.jp}
\affiliation{Department of Applied Physics, University of Tokyo, 7-3-1 Hongo, Bunkyo-ku, Tokyo 113-8656, Japan}
\affiliation{Theoretical Quantum Physics Laboratory, RIKEN Cluster for Pioneering Research (CPR), Wako-shi, Saitama 351-0198, Japan}
\affiliation{JST, PRESTO, 4-1-8 Honcho, Kawaguchi, Saitama, 332-0012, Japan}

\begin{abstract}
    We present a unified approach to analyzing the cost of various quantum error mitigation methods on the basis of quantum estimation theory.
    By analyzing the quantum Fisher information matrix of a virtual quantum circuit that effectively represents the operations of quantum error mitigation methods, we derive for a generic layered quantum circuit under a wide class of Markovian noise that, unbiased estimation of an observable encounters an exponential growth with the circuit depth in the lower bound on the measurement cost.
    Under the global depolarizing noise, we in particular find that the bound can be asymptotically saturated by merely rescaling the measurement results.
    Moreover, we prove for random circuits with local noise that the cost grows exponentially also with the qubit count.
    Our numerical simulations support the observation that, even if the circuit has only linear connectivity, such as the brick-wall structure, each noise channel converges to the global depolarizing channel with its strength growing exponentially with the qubit count.
    This not only implies the exponential growth of cost both with the depth and qubit count, but also validates the rescaling technique for sufficiently deep quantum circuits.
    Our results contribute to the understanding of the physical limitations of quantum error mitigation and offer a new criterion for evaluating the performance of quantum error mitigation techniques.
\end{abstract}
	
\maketitle

\emph{Introduction.---}
One of the central problems in quantum technology is to establish control and understanding of unwanted noise, since an accumulation of errors may eventually spoil the practical advantage of quantum devices.
In the case of quantum computing, an elegant framework of quantum error correction has been developed as a fundamental countermeasure~\cite{shor1995scheme,knill1996threshold, aharanov2008fault, lidar2013quantum, nielsen2002quantum, ofek2016extending, krinner2022realizing, zhao2022realization}, while it remains years to decades ahead when we can reliably implement provably advantageous quantum algorithms.
A realistic and powerful alternative for near-future devices is to employ the art of quantum error mitigation (QEM); instead of consuming an excessive number of qubits to {\it correct} the bias caused by noise via interleaved measurement and feedback, we aim to {\it mitigate} their effect via appropriate post-processing in trade of an increased number of measurements.

A wide variety of QEM methods have been proposed: zero-noise extrapolation~\cite{li2017efficient,temme2017error,kandala2019error}, probabilistic error cancellation~\cite{temme2017error, endo2018practical, berg2022probabilistic}, virtual distillation~\cite{huggins_virtual_2021, koczor_exponential_2021, huo2022dualstate, czarnik2021qubit}, (generalized) quantum subspace expansion~\cite{mcclean_2017, mcclean2020decoding, yoshioka2022variational, yoshioka2022generalized}, symmetry verification/expansion~\cite{bonet-monroig2018lowcost, mcardle2019error, cai2021quantumerror}, and learning-based error mitigation~\cite{Czarnik2021errormitigation, strikis2021learning}, to name a few~(Refer to Ref.~\cite{endo2021hybrid, cai2022quantum} for review). 
The growing number of demonstrations by both numerical and experimental means shows that the QEM has become vital~\cite{kandala2019error,sagastizabal2019experimental, sun2021mitigating, zhang2020error}.
Meanwhile, there are so far only a few guiding principles to choose from existing QEM methods~\cite{takagi2022fundamental, wang2021can}, due to the limited theoretical understanding of their fundamental aspects.
It is an urgent task to understand what is the limit of QEM, in particular, the required  resource to recover the desired quantum circuit output.

We find that quantum estimation theory provides a powerful tool to address this problem.
Quantum estimation theory claims that, given an unbiased estimator of a physical observable, its estimation uncertainty can be characterized by the quantum Fisher information~\cite{helstrom1969quantum, holevo2011probabilistic, hayashi2006quantum}.
For example, the sampling cost for constructing an unbiased estimator for \emph{noiseless} quantum states from measurements in \emph{noisy} quantum states can be bounded using the quantum Fisher information~\cite{watanabe2010optimal}.
While this strongly implies that the quantum estimation theory yields a tool to analyze the trade-off cost \black{to recover the} desired quantum operation, it has remained totally unknown how to investigate realistic computation models such as quantum circuits, in which the \black{holistic effect of the error} cannot be expressed by a single noise channel in general.
Moreover, since QEM methods are mostly not purely classical post-processing but also require additional quantum operations, the existing framework is not straightforwardly applicable.

In this Letter, we aim to fill these gaps by extending the applicability of quantum estimation theory.
By analyzing the quantum Fisher information matrix of an enlarged virtual quantum circuit which translates the operations of QEM methods, we show that the lower bound of the sampling cost for unbiased QEM grows exponentially with the circuit depth $L$ for a generic layered quantum circuit under a wide class of noise~(Theorem~\ref{thm_1}). Furthermore, for random layered circuits under local noise, we show that the cost grows exponentially also with the qubit count $n$~(Theorem~\ref{thm_2}).
We have also numerically verified that noise channels in the large depth regime may be effectively described by the global depolarizing channel whose strength grows exponentially with $n$, for which we provide an optimal technique to suppress the effect of noise.
\red{These results surpass some prior work suggesting some exponential growth (not necessarily the sampling cost of QEM) under the local depolarizing noise~\cite{aharonov1996limitations, takagi2022fundamental, wang2021can} from both theoretical and practical points of view:}
our result not only provides the first mathematical proof for a necessary condition for unbiased QEM under a wide range of noise, but also provides practical guidelines toward cost-optimal QEM.



\emph{Problem setup.---}
Analysis of sample complexity via the quantum estimation theory assumes operations to be expressed as quantum channels. Therefore, it is beneficial to embed QEM operations into a quantum circuit.
Below, we first define a noiseless and noisy layered quantum circuit, and then present the concept of a virtual quantum circuit that encodes QEM operations.

Let $\hat{\rho} = \mathcal{U}_L \circ \cdots \circ \mathcal{U}_1(\hat{\rho}_0)$ ($\mathcal{U}_l(\cdot) = \hat{U}_l\cdot \hat{U}_l^\dagger$) be an unknown \red{$n$-qubit} target state generated from $L$ layers of noiseless unitary gates $\{\mathcal{U}_l\}_{l=1}^L$ operating on an initial state $\hat{\rho}_0$. The target state $\hat{\rho}$ can be
 parameterized by the generalized Bloch vector~\cite{kimura2003bloch} $\vb*{\theta} \in \mathbb{R}^{4^{n}-1}$ as
\begin{eqnarray}
    \hat{\rho} = \frac{1}{2^n} \hat{I}  + 2^{(-1-n)/2}\vb*{\theta} \cdot \hat{\vb*{P}},
\end{eqnarray}
where $\hat{I} \equiv \hat{\sigma}_0^{\otimes n}$ and $\hat{\vb*{P}} = \{\hat{P}_i\}_{i=1}^{2^{2n}-1}$ is an array of non-trivial tensor product of Pauli operators $\hat{P}_i \in \qty{\hat{\sigma}_0, \hat{\sigma}_x, \hat{\sigma}_y, \hat{\sigma}_z}^{\otimes n} \setminus \qty{\hat{\sigma}_0}^{\otimes n}$.

An $n$-qubit noisy layered circuit is defined to have the following structure: (i) noiseless preparation of initial state $\hat{\rho}_0$~\cite{Note2}, (ii) $L$ layers of noisy unitary operations $\qty{\mathcal{E}_l\circ\mathcal{U}_l}_{l = 1}^L$ with $\mathcal{E}_l$ assumed to be a Markovian error, and (iii) noiseless POVM measurement$~\mathcal{M}_0$ aimed to estimate the expectation value of a traceless observable  $\hat{X} = \vb*{x}\cdot\hat{\vb*{P}}$ with $\vb*{x} \in \mathbb{R}^{4^n-1}$.
Each noise channel $\mathcal{E}_l$ maps a generalized Bloch vector as 
\begin{eqnarray}
\mathcal{E}_l : \vb*{\theta} \mapsto A_{l}\vb*{\theta} + \vb*{c}_{l},    
\end{eqnarray}
where $(A_{l})_{ij}=2^{-n}\mathrm{tr}[\hat{P}_i\mathcal{E}_l(\hat{P}_j)]$ is the unital part of the Pauli transfer matrix of $\mathcal{E}_{l}$ and $(\vb*{c}_l)_i =2^{(1-3n)/2}\mathrm{tr}[\hat{P}_i\mathcal{E}_l(\hat{I})]$ quantifies the non-unital action of the noise~\cite{watanabe2010optimal}.
We also define noise strength $\Gamma(\mathcal{E}_{l}) \equiv \norm{A_{l}}^{-1}$ with $\norm{A_{l}} = \max_{\vb*{e}\in \mathbb{R}^{2^{2n}-1}} \frac{\norm{A_{l}\vb*{e}}}{\norm{\vb*{e}}}$ and $\norm{\vb*{e}} = \sum_i\abs{e_i}^2$, which represents the minimal degree of shrinkage of the generalized Bloch sphere caused by $\mathcal{E}_{l}$, and $\gamma = \min_l\qty{\Gamma(\mathcal{E}_{l})}_l$ as the minimal strength among the different noise.

The objective of QEM methods is to remove the effect of the noise channels $\qty{\mathcal{E}_l}_{l=1}^L$ so that we have an unbiased estimator of traceless observable $\hat{X}$, or namely $\ev*{\hat{X}} \equiv \mathrm{tr}[\hat{\rho}\hat{X}] = 2^{(n-1)/2}\vb*{\theta}\cdot\vb*{x}$.
Since the essence of QEM is to run noisy quantum circuits with implementable modifications into the gates, errors, and classical postprocessing, we can construct a {\it virtual} quantum circuit which encompasses the functionality of QEM methods.

As is shown in Fig.~\ref{fig_mitigation_circuit}, 
the virtual circuit involves $N$ copies of noisy layered circuits with three-fold modifications from the original one:
 (boosted) noise $\mathcal{E}_{lm}$ in the $l$-th layer of $m$-th copy such that $\Gamma(\mathcal{E}_{lm})\geq\Gamma(\mathcal{E}_{l})\geq \gamma$,  classical register $\hat{\rho}_{\mathrm{c}, m}$ coupled with the system qubits via the additional operation $\mathcal{C}_{lm}$, and finally, the POVM measurement $\mathcal{M}$ performed on the entire copies to output the estimator of $\ev*{\hat{X}}$.
Classical register ${\hat\rho}_{c, m}$ is initialized with probabilistic mixtures of computational bases as $\hat{\rho}_{\mathrm{c}, m} = \sum_{i}p_{mi}\ketbra{i}$,
and additional operation $\mathcal{C}_{lm}$ performs unitary operation $\mathcal{C}_{lmi}$ according to the state of the classical register as $\mathcal{C}_{lm} = \sum_{i} \mathcal{C}_{lmi}\otimes\ketbra{i}$.
\red{Note that the virtual circuit structure excludes the quantum error correction. This is because we only allow $\mathcal{C}_{lm}$ to be unitary operation according to the state of the classical registers.}
We further describe in SM how various QEM methods can be mapped into this virtual circuit structure~\cite{Note1}.

\begin{figure}[t]
    \begin{center}
        \includegraphics[width=0.9\linewidth]{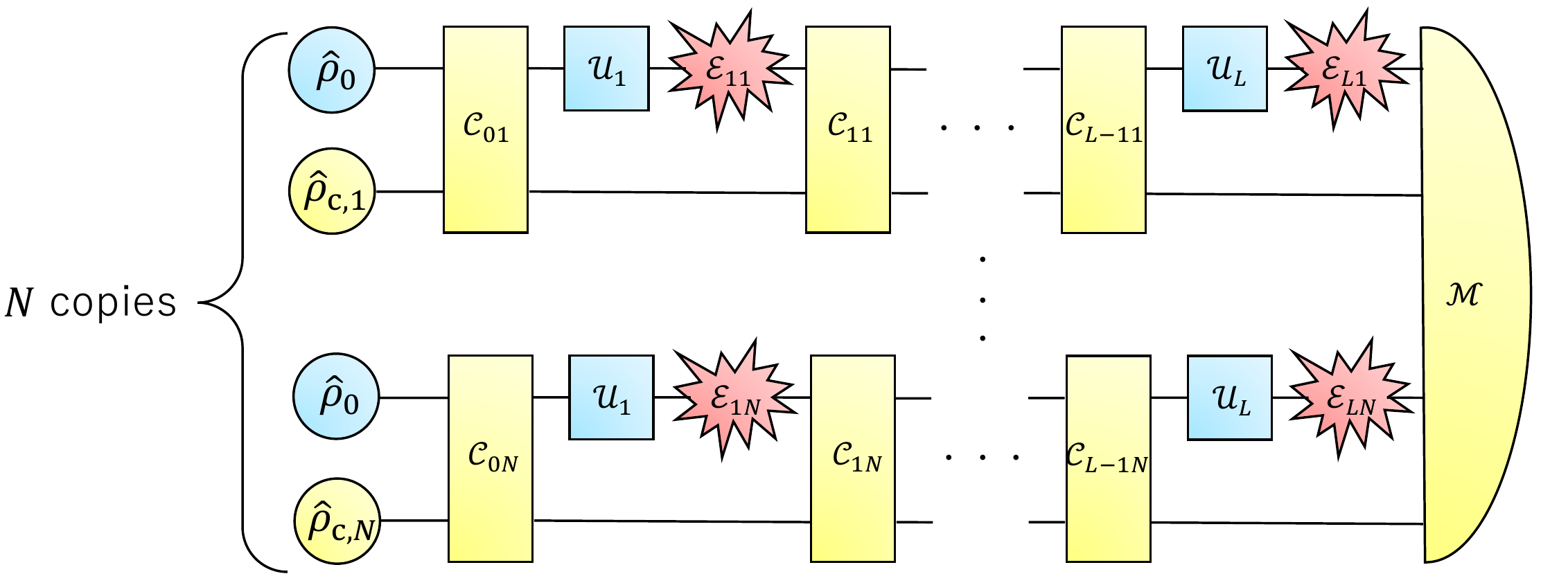}
        \caption{A virtual quantum circuit structure that gives an equivalent representation of most existing QEM methods.
        \black{The blue (red) coloring denotes that the operation is ideal  (noisy), while the operations with yellow coloring explicitly involve QEM operations.}
        }
        \label{fig_mitigation_circuit}
    \end{center}
\end{figure}


The cost of QEM can be defined as the number of copies $N$ of the noisy circuits, or the sample complexity, which roughly can be interpreted as the number of measurements on the actual setup.
Our goal is to derive the lower bound on the cost $N$ required to perform unbiased estimation of $\ev*{\hat{X}}$, by analyzing the evolution of the quantum Fisher information matrix of quantum states generated by the virtual circuit.
Note that the lower bound described below also holds even when we think of measurement error and measurement error mitigation~\cite{kandala2017hardware, heinsoo2018rapid, bravyi2020mitigating}.
This is because noisy measurement followed by the process of measurement error mitigation can be seen as a single POVM measurement.

\emph{Main Results.---}
In order to achieve our goal, we re-express the $m$-th copy of the quantum state in the virtual circuit as
\begin{eqnarray}
    \mathcal{E}'_{m}(\hat{\rho}(\vb*{\theta})\otimes\hat{\rho}_{\mathrm{c}, m}),
\end{eqnarray}
where $\mathcal{E}'_{m}$ is an {\it effective} noise channel defined by compiling all the gates as  $\mathcal{E}'_{m} = \mathcal{E}_{Lm}\circ\mathcal{U}_L\circ\mathcal{C}_{L-1m}\circ\cdots\circ\mathcal{E}_{1m}\circ\mathcal{U}_1\circ\mathcal{C}_{0m}\circ\mathcal{U}_1^{-1}\circ\cdots\circ\mathcal{U}_L^{-1}$.
This compilation allows us to calculate the quantum Fisher information matrix of the state right before the measurement.
To be concrete, we analyze the SLD Fisher information matrix $J$~\cite{helstrom1967minimum} of the quantum state $\bigotimes_{m=1}^N\mathcal{E}'_{m}(\hat{\rho}(\vb*{\theta})\otimes\hat{\rho}_{\mathrm{c}, m})$.


\begin{figure*}[t]
    \begin{center}
        \resizebox{0.9\hsize}{!}{\includegraphics{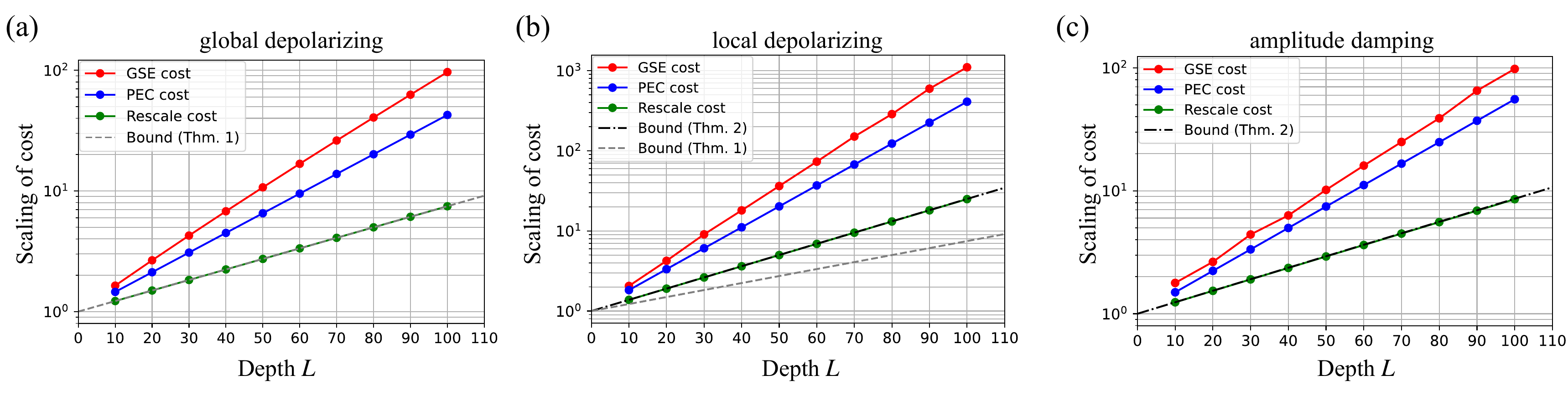}}
        \caption{Scaling of the cost to perform QEM methods for random Clifford circuit of $n=2$ qubits under (a) global depolarizing noise, (b) local depolarizing noise, and (c) local amplitude damping noise with error rate $p=0.01$.
        The red, blue, and green lines denote the sampling overhead of generalized subspace expansion~\cite{yoshioka2022generalized, yang2023dual} using power subspace, the probabilistic error cancellation as derived in Ref.~\cite{takagi2020optimal}, and the rescaling technique as explained in the main text. The rescaling factor is $(1-p)^{-L}$ and $(1-p)^{-3nL4^{n-1}/(4^n-1)}$ for global and local depolarizing noise, and $(1-p)^{-2nL4^{n-1}/(4^n-1)}$ for amplitude damping noise, respectively.
        Bound (Thm. 1) and Bound (Thm. 2) represent the lower bound of the cost obtained from Theorem 1 and Theorem 2, respectively.
        The explicit scaling of Bound (Thm.2) is given in Eq. (8).
        Note that GSE and the rescaling methods do not completely eliminate the errors for (b) and (c), while we confirm a significant reduction of bias. 
        }
        \label{fig_numerics}
    \end{center}
\end{figure*}

We find that $J$ can be bounded as
$J \lesssim \sum_m\Gamma(\mathcal{E}_m')^{-2} I $$\lesssim N\gamma^{-2L}$, which implicates the exponential decay of $J$ with the circuit depth $L$.
By combining this fact with the quantum Cram\'{e}r-Rao inequality, which relates $J$ with the standard deviation $\varepsilon$ of an unbiased estimator~\cite{braunstein1994statistical}, we immediately obtain the following theorem for the cost $N$ of the unbiased QEM (See SM~\cite{Note1} for the proof):
\begin{thm}
    \label{thm_1}
    Suppose that the noise $\mathcal{E}_{lm}$ satisfies the following conditions for all $l$ and $m$:
    \vspace{-5pt}
    \begin{itemize}
        \setlength{\parskip}{0cm} 
        \setlength{\itemsep}{0cm}
        \item[(I)] For all $\hat{\rho} \neq \hat{\sigma}$, $\mathcal{E}_{lm}(\hat{\rho})\neq\mathcal{E}_{lm}(\hat{\sigma})$.
        \item[(II)] For all $\hat{\rho}$, $\mathcal{E}_{lm}(\hat{\rho})$ is full rank, that is, $\mathcal{E}_{lm}(\hat{\rho})$ is a positive definite matrix whose eigenvalues are all greater than zero.
    \end{itemize}
    \vspace{-5pt}
    Then, the cost $N$ required for any unbiased estimator of $\langle \hat{X}\rangle$ with standard deviation $\varepsilon$ constructed from QEM that can be translated into the virtual quantum circuit in Fig. \ref{fig_mitigation_circuit} satisfies
    \begin{eqnarray}
        N \geq \frac{\norm{\vb*{x}}^2}{\varepsilon^2}\beta\gamma^{2L},\label{eq:thm_1}
    \end{eqnarray}
    where $\beta$ is the largest $0 < \beta < 1$ such that $\mathcal{E}_{lm}(\hat{\rho}) - \frac{\beta}{2^n}\hat{I} \geq 0$ for all $\hat{\rho}$, $l$, and $m$.
    Suppose further that the noise $\mathcal{E}_{lm}$ is unital, that is, $\mathcal{E}_{lm}(\frac{\hat{I}}{2^n}) = \frac{\hat{I}}{2^n}$.
    Then, the cost $N$ satisfies
    \begin{eqnarray}
        N \geq \frac{\norm{\vb*{x}}^2}{\varepsilon^2}\qty(1-(1-\beta)^L)\gamma^{2L} \sim \frac{\norm{\vb*{x}}^2}{\varepsilon^2}\gamma^{2L}.
    \end{eqnarray}
\end{thm}
Theorem \ref{thm_1} shows that, if $\gamma > 1$, the cost $N$ of the unbiased QEM grows exponentially with the circuit depth $L$
no matter how we choose $\mathcal{E}_{lm}$ (with their strength bounded from below), $\mathcal{C}_{lm}$, and $\mathcal{M}$.
We can indeed show $\gamma > 1$ for unital noise under the condition (II) (See SM~\cite{Note1} for details).

Let us make a few remarks on the condition of Theorem~\ref{thm_1}. Condition (I) is a necessary condition for successful QEM, meaning that the information of the quantum state is not completely destroyed by noise.
Condition (II)  means that the variance of all observable of any state after the noise is applied is non-zero.
In other words, for any observable and quantum state, the cost of obtaining an unbiased estimator from the measurement of the noisy state is greater than zero.
We also remark that $\beta$ is a constant that represents how far away the generalized Bloch sphere is from the original surface due to the noise.

It is noteworthy that the lower bound stated in Theorem 1 is for a generic layered quantum circuit. Since it also involves circuits that only weakly entangle qubits, the lower bound~\eqref{eq:thm_1} does not depend on the qubit count $n$. However, if the quantum circuit scrambles the quantum state strong enough, we expect that every noise affects the measurement outcome; we must pay overhead to eliminate every local noise and thus encounter dependence on $n$. In fact, under local noise we can tighten the bound as in the following informal theorem (See SM for details~\cite{Note1}):
\begin{thm}
    \label{thm_2}
    Let $U_1, U_2, ..., U_{L-1}, U_L$ be $n$-qubit unitary gate drawn from a set of random unitary that form unitary 2-design~\cite{dankert2009exact} and $\mathcal{E}_l$ be a local noise. 
    Then, 
    there is exponential growth with both qubit count $n$ and depth $L$ in the average over the number of copies $N$ 
    required to perform unbiased estimation of $\ev*{\hat{X}}$ over $\{U_1, ..., U_L\}$.
\end{thm}

\emph{Applications.---}
Here, we compare the obtained bounds and the practical performance of QEM methods under realistic noise channels to determine the efficiency of existing methods.
For the sake of illustrativeness,
we consider three typical noise channels: the global and local depolarizing noise as representative of unital noise, and amplitude damping noise as representative of nonunital noise.


First, we consider the case where all unitary gates are followed by the global depolarizing noise $\mathcal{E}_{lm} : \vb*{\theta} \mapsto (1-p_{lm})\vb*{\theta}$, where the error rate is lower bounded as $p_{lm} \geq p$.
Since the global depolarizing noise channel is unital and satisfies the assumptions of Theorem \ref{thm_1} with minimal noise strength $\gamma = \frac{1}{1-p}$ and $\beta = p$, the cost $N$ required for the unbiased estimator of the expectation value $\ev*{\hat{X}}$ constructed from QEM shall satisfy
\begin{eqnarray}
    N &\geq& \frac{\norm{\vb*{x}}^2}{\varepsilon^2}\qty(1-(1-p)^L)\qty(\frac{1}{1-p})^{2L}\\
    \label{eq_cost_gdn}
    &\sim& \frac{\norm{\vb*{x}}^2}{\varepsilon^2}\qty(\frac{1}{1-p})^{2L}.
\end{eqnarray}

We can show that Eq.~\eqref{eq_cost_gdn} can be saturated in the limit of large $L$.
By setting $p_{lm} = p$ and ignoring the classical registers and the additional operations, the effective noise channel $\mathcal{E}_j'$ can be seen as the global depolarizing noise channel with error rate $1-(1-p)^L$.
Since the measurement on the observable $\hat{X}$ yields $\ev*{\hat{X}} ^{\mathrm{noisy}} = 2^{(n-1)/2}(1-p)^L\vb*{\theta}\cdot\vb*{x}$, we achieve unbiased estimation by rescaling the measurement result as $(1-p)^{-L}\ev*{\hat{X}} ^{\mathrm{noisy}}$.
Since the estimation variance on $\ev*{\hat{X}}^{\rm noisy}$ is $\norm{\vb*{x}}^2$ in the limit of large $L$, the sampling cost to estimate $\ev*{\hat{X}}$ approaches 
$\frac{\norm{\vb*{x}}^2}{\varepsilon^2}\qty(\frac{1}{1-p})^{2L}$, which satisfies the lower bound of Theorem \ref{thm_1}.
We compare these results in Fig. \ref{fig_numerics} (a) with other error mitigation methods that also allow unbiased estimation.

Next, we consider the case of local noise $\mathcal{E}_{lm} = (\mathcal{E}^{(0)}_{lm})^{\otimes n}$ 
with $\mathcal{E}^{(0)}_{lm}:\vb*{\theta}\mapsto (1-p_{lm})\vb*{\theta}$ for local depolarizing and $(\theta_x,\theta_y,\theta_z)\mapsto(\sqrt{1-p_{lm}}\theta_x,\sqrt{1-p_{lm}}\theta_y,(1-p_{lm})\theta_z+p_{lm})$ for amplitude damping noise, where $\vb*{\theta} = (\theta_x,\theta_y,\theta_z)$ denotes the Bloch vector and the error rate is lower bounded as $p_{lm} \geq p$.
From Theorem \ref{thm_1}, we can show that the cost $N$ required by any unbiased estimator of the expectation value $\ev*{\hat{X}}$ constructed from QEM satisfies Eq. (\ref{eq_cost_gdn}) in the case of local depolarizing noise.
For a random circuit whose unitary gate is drawn from unitary 2-design such as $n$-qubit Clifford group~\cite{dankert2009exact}, we can even tighten this bound in the average case as
\begin{eqnarray}
    \label{eq_average_localdep}
    \mathbb{E}[N] \geq 
    \begin{cases}
    O\qty(\qty(1 + \frac{3}{2}\frac{4^n}{4^n-1} p)^{nL}) & \text{(local dep.)}\\
    O\qty(\qty(1+\frac{4^n}{4^n-1}p)^{nL}) & \text{(amp. damping)}
    \end{cases}
\end{eqnarray}
from Theorem \ref{thm_2}.
We compare these results in Fig. \ref{fig_numerics} (b)(c) with some QEM methods.

While the scaling of Eq.~\eqref{eq_average_localdep} is derived under the assumption of unitary 2-design, our numerical simulation suggests that the bound shall hold for even wider class of quantum circuits.
Concretely, as is presented in Fig.~\ref{fig_noise_convergence}, the effect of each noise becomes indiscriminable from that of the global depolarizing noise whose error rate grows exponentially with $n$ in the large-$L$ regime, even when any of $\{\mathcal{U}_l\}$ does not constitute unitary 2-design.
These results are in agreement with the phenomenological argument provided in Ref.~\cite{qin2021error} that, noise in deep layered circuits shall be modeled by global depolarizing noise with its strength fluctuating as $O(1/\sqrt{L})$. 
These facts not only give us another evidence for scaling as in Eq. (\ref{eq_average_localdep}) but also imply that, although we cannot remove bias completely, we may optimally suppress the effect of noise by just rescaling the measurement results as in the case of global depolarizing noise. 
We also applied our results for local dephasing noise, and showed that such a picture also holds as well~(See SM for details~\cite{Note1}).

\begin{figure}[t]
    \begin{center}
        \resizebox{0.7\hsize}{!}{\includegraphics{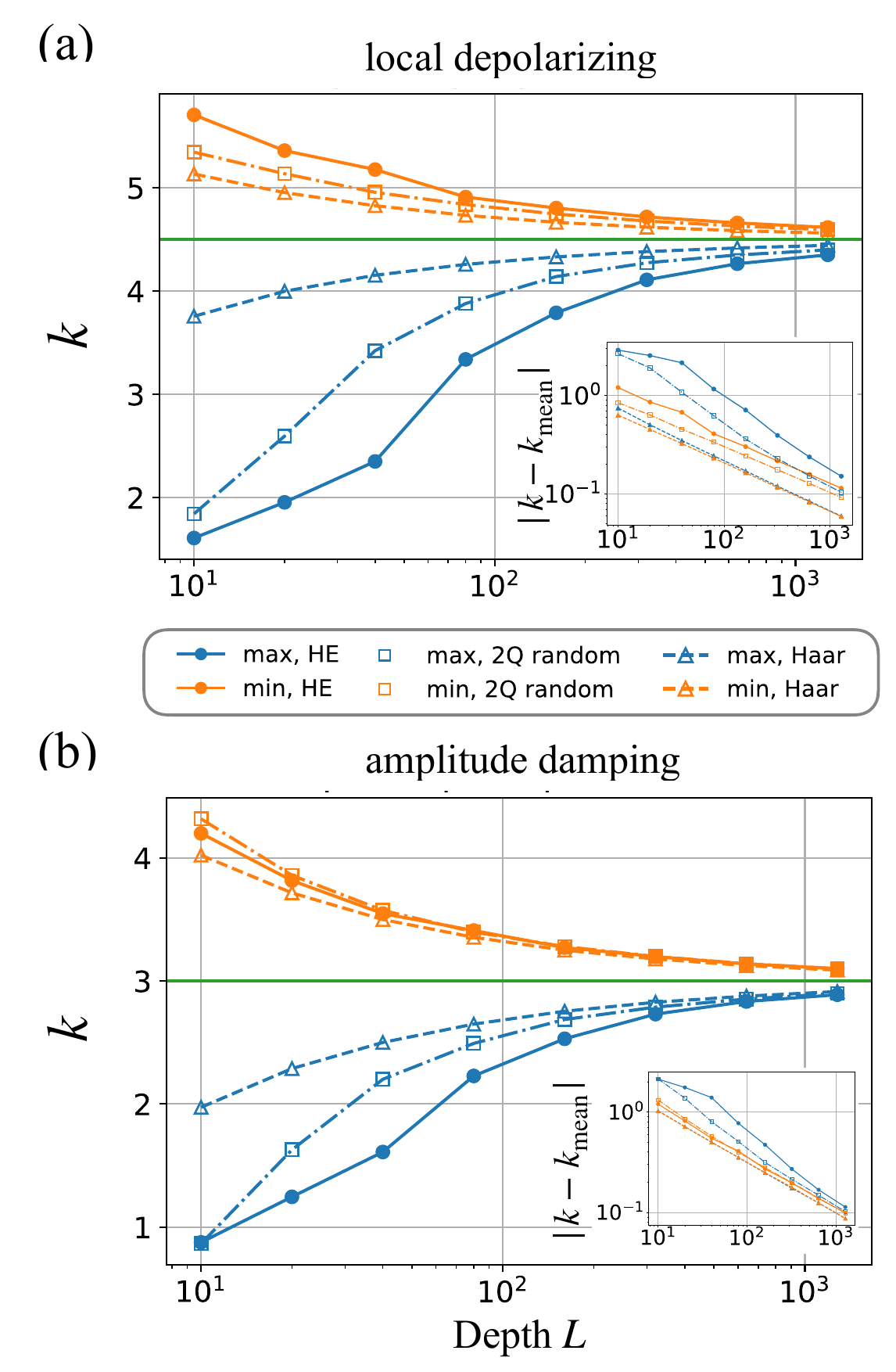}}
        \caption{Convergence of (a) local depolarizing and (b) amplitude damping into global depolarizing noise under random circuits of $n=6$ qubits with error rate $p=0.0001$.
        We denote by $(1-p)^{kL}$ the singular values of the unital part of the Pauli transfer matrix for the effective noise channel $\mathcal{E}_m'$ at each depth $L$, where $k$ for the maximal and minimal ones are plotted in this figure. As is highlighted in the inset, we find that all $k$'s approach the geometric mean $k_{\rm mean}$ of the singular values \red{for each noise channel} with its fluctuation scaling as $O(1/\sqrt{L})$, implying the convergence to the global depolarizing noise.
        For instance, $k_{\rm mean} = 3n4^{n-1}/(4^n-1)$ for local depolarizing and $k_{\rm mean} = 2n4^{n-1}/(4^n-1)$ for amplitude damping.
        Here, we consider three class of random circuits: hardware-efficient ansatz with random parameters, 2-qubit random unitary between random pairs, and Haar random unitary (See SM for details~\cite{Note1}).
        }
        \label{fig_noise_convergence}
    \end{center}
\end{figure}

\emph{Conclusion.---}
In this Letter,  we have presented a theoretical analysis of quantum error mitigation (QEM)
to reveal two unavoidable cost bound for unbiased QEM: exponential growth with depth $L$ for generic layered quantum circuits, and furthermore exponential growth with qubit count $n$ for random/chaotic quantum circuits.
The lower bound is shown to be saturated under global depolarizing noise by just rescaling the measurement result, while numerical results suggest that other noise may also be mitigated as well when the circuit is sufficiently deep, since the noise including both unital ones and nonunital ones may converge to the global depolarizing noise.


We envision a rich variety of future directions. Here we mention the most important two in order.
\red{
The first is to develop even more knowledge of cost-optimal QEM, especially in the early fault-tolerant regime.
Even for the fault-tolerant quantum computer, a slight amount of logical errors may remain in the circuit (especially in the early regime).
The implemented quantum circuits will be much deeper than those of NISQ, and thus the convergence of logical errors to global depolarizing noise is expected to be stronger.
Thus, we believe that we can use our results to develop ways to utilize long-term quantum computation in the most efficient way.
}

The second is to incorporate the influence of bias in the estimators.
QEM methods in reality are not designed to completely remove the effect of the noise,
and a slight bias is allowed to remain in the estimation results.
In such situations, we can expect a trade-off relationship between the cost, bias, and uncertainty of the estimator.
Extending the results on single parameter estimation~\cite{liu2016valid} is left as an interesting future work.

\ \\
\emph{Acknowledgements.---}
The authors wish to thank Ryuji Takagi, Hiroyasu Tajima, and Mile Gu for insightful discussions and for sharing a preliminary version of their manuscript. The authors are also grateful to fruitful discussions with Sergey Bravyi, Suguru Endo, Keisuke Fujii, Liang Jiang, Yosuke Mitsuhashi, Changhun Oh, \red{Zlatko Minev}, Kunal Sharma, Yasunari Suzuki, and  Kristan Temme.
This work was supported by JST ERATO-FS Grant Number JPMJER2204, JST Grant Number JPMJPF2221, \red{JST CREST Grant Number JPMJCR23I4, and JST ERATO Grant Number  JPMJER2302, Japan}.
K.T. is supported by Worldleading Innovative Graduate Study Program for Materials Research, Industry, and Technology (MERITWINGS) of the University of Tokyo.
T.S. is supported by JSPS KAKENHI Grant Number JP19H05796, Japan.
N.Y. wishes to
thank JST PRESTO No. JPMJPR2119.
T. S. and N.Y. acknowledge the support from IBM Quantum.

\ \\
\emph{Note added.---}
During the completion of our manuscript, we
became aware of an independent work by Takagi et al.~\cite{takagi2022universal}, which
also showed the exponential growth of the cost $N$ with circuit depth based on analysis of discriminability between quantum states.
Also, Quek et al.~\cite{quek2022exponentially} has theoretically analyzed the exponential scaling of sample complexity regarding both qubit counts and circuit depth via statistical learning theory.
\red{We note that, for non-unital noise, our average bound is quadratically tighter than the bound obtained by Refs.~\cite{quek2022exponentially}.}
\red{See SM~\cite{Note1} for more details.}

\let\oldaddcontentsline\addcontentsline
\renewcommand{\addcontentsline}[3]{}
\bibliography{bib.bib}
\let\addcontentsline\oldaddcontentsline
\onecolumngrid

\clearpage
\begin{center}
	\Large
	\textbf{Supplementary Materials for: Universal cost bound of quantum error mitigation based on quantum estimation theory}
\end{center}

\setcounter{section}{0}
\setcounter{equation}{0}
\setcounter{figure}{0}
\setcounter{table}{0}
\setcounter{thm}{0}
\renewcommand{\thesection}{S\arabic{section}}
\renewcommand{\theequation}{S\arabic{equation}}
\renewcommand{\thefigure}{S\arabic{figure}}
\renewcommand{\thetable}{S\arabic{table}}
\renewcommand{\thethm}{S\arabic{thm}}
\renewcommand{\thelem}{S\arabic{lem}}

\addtocontents{toc}{\protect\setcounter{tocdepth}{0}}
{
\hypersetup{linkcolor=blue}
\tableofcontents
}

\section{A brief review of quantum estimation theory}
In this section, we present a brief review of quantum estimation theory~\cite{helstrom1969quantum, holevo2011probabilistic, hayashi2006quantum}.
Quantum estimation theory characterizes the amount of information that can be extracted from quantum measurement. 
Specifically, an information-theoretic quantity called the quantum Fisher information~\cite{helstrom1967minimum} is known to describe the uncertainty (or variance) of unbiased estimation on a physical property of an unknown quantum state.
In the following, we aim to introduce the mathematical tools of the quantum estimation theory by reviewing their main claims under two setups. The first is the case when we have access to the noiseless quantum state, and the second is when the unknown quantum state is exposed to a single known error channel.

Let us assume that the unknown noiseless quantum state is parameterized as $\hat{\rho}(\vb*{\theta})$ where $\vb*{\theta}$ are unknown parameters, and also that the quantity we wish to estimate is expressed as a function of $\vb*{\theta}$ as $f(\vb*{\theta})$.
The question is, given independent many copies of parameterized quantum states $\hat{\rho}(\vb*{\theta})$, how accurately the value of $f(\vb*{\theta})$ can be estimated by performing POVM measurements on these states.
We emphasize that such a formulation encompasses versatile problems. For instance, the goal of quantum metrology is to estimate $f(\vb*{\theta}) = \vb*{\theta}$ where $\vb*{\theta}$ denotes the amplitude of target field imposed on the system. 
\black{As we shortly explain, it is also compatible with one of the essential tasks in the context of quantum computing, namely to identify $f(\vb*{\theta})$ with the expectation value of a physical observable that one wishes to estimate.}

According to the quantum estimation theory, the quantum Fisher information matrix characterizes the uncertainty of the estimator~(See Fig.~\ref{fig_quantum_estimation_thoery}).
To be concrete, when the estimator $\qty(f(\vb*{\theta}))^{\mathrm{est}}$ of an unknown $f(\vb*{\theta})$ is unbiased, i.e., when expectation value of $\qty(f(\vb*{\theta}))^{\mathrm{est}}$ equals $f(\vb*{\theta})$ \red{for all $\vb*{\theta}$}, the variance of the estimator is related with the number of copies $N$ via the quantum Cram\'{e}r-Rao inequality~\cite{braunstein1994statistical} as
\begin{eqnarray}
    \label{eq_QCRB}
    \mathrm{Var}(\qty(f(\vb*{\theta}))^{\mathrm{est}}) \geq \frac{1}{N} \grad_{\vb*{\theta}}(f(\vb*{\theta}))^T J(\hat{\rho}(\vb*{\theta}))^{-1} \grad_{\vb*{\theta}}(f(\vb*{\theta})),
\end{eqnarray}
where $\grad_{\vb*{\theta}}(f(\vb*{\theta}))^T J(\hat{\rho}(\vb*{\theta}))^{-1} \grad_{\vb*{\theta}}(f(\vb*{\theta}))$ in the right-hand side denotes the inverse of 
the quantum Fisher information of $f(\vb*{\theta})$.
Here, the quantum Fisher information matrix $J$ of $\hat{\rho}(\vb*{\theta})$ is defined as follows,
\begin{eqnarray}
    \qty[J]_{ij} = \frac{1}{2}\tr\qty[\hat{\rho}\qty{\hat{L}_i, \hat{L}_j}],
\end{eqnarray}
where $\qty{\cdot,\cdot}$ represents the anti-commutation and $\hat{L}_i$ is the symmetric logarithmic
derivative (SLD) operator defined as
\begin{eqnarray}
    \pdv{\hat{\rho}}{\theta_i} = \frac{1}{2}\qty{\hat{\rho}, \hat{L}_i}.
\end{eqnarray}

\begin{figure}[ht]
    \begin{center}
        \includegraphics[width=120mm]{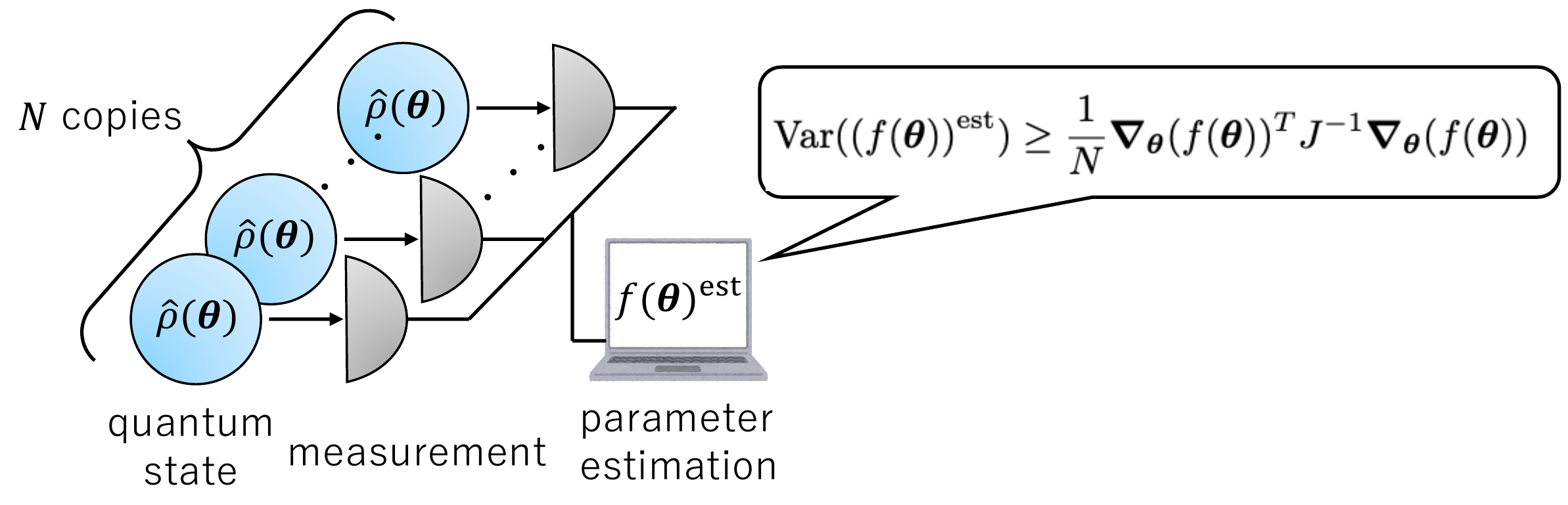}
        \caption{Overview of quantum estimation theory. $N$ independent copies of parameterized quantum states $\hat{\rho}(\vb*{\theta})$ are measured to obtain an unbiased estimator $\qty(f(\vb*{\theta}))^{\mathrm{est}}$ of an unknown $f(\vb*{\theta})$. The variance of the estimator satisfies the Cram\'{e}r-Rao inequality. In this figure, the POVM measurements are performed independently on each quantum state, but the quantum Cram\'{e}r-Rao inequality also holds even when POVM measurement is an entangled measurement over the entire copies of $\hat{\rho}(\vb*{\theta})$.}
        \label{fig_quantum_estimation_thoery}
    \end{center}
\end{figure}

Let us next see how the explicit form of the right-hand side of the inequality~\eqref{eq_QCRB} is given under the task of estimating physical observables.
In this work, we mainly consider the quantum states of $n$-qubit systems which are parameterized by the generalized Bloch vector \cite{kimura2003bloch}.
A quantum state $\hat{\rho}$ of an $n$-qubit system and traceless observable $\hat{X}$ can be expressed as
\begin{eqnarray}
    \hat{\rho}(\vb*{\theta}) &=& \frac{1}{2^n} \hat{I}  + 2^{(-1-n)/2}\vb*{\theta} \cdot \hat{\vb*{P}}\left( = \frac{1}{2^n} \hat{I}  + \frac{1}{2}\vb*{\theta} \cdot \hat{\vb*{\lambda}}\right),\\
    \hat{X} &=& \vb*{x}\cdot\hat{\vb*{P}}, 
\end{eqnarray}
where $\vb*{\theta} \in \mathbb{R}^{2^{2n}-1}$ is the generalized Bloch vector, $\vb*{x} \in \mathbb{R}^{4^n-1}$, $\hat{\vb*{P}} = \{\hat{P}_i\}_{i=1}^{2^{2n}-1}$ is an array of nontrivial Pauli string, i.e., $\hat{P}_i \in \qty{\hat{\sigma}_0, \hat{\sigma}_x, \hat{\sigma}_y, \hat{\sigma}_z}^{\otimes n} \setminus \qty{\hat{\sigma}_0}^{\otimes n}$ is a tensor product of Pauli operators $\hat{\sigma}_x, \hat{\sigma}_y, \hat{\sigma}_z$ and the 1-qubit unit operator $\hat{\sigma}_0$, and $\hat{I} = \hat{\sigma}_0^{\otimes n}$.
Note that these are related to generators of the Lie algebra $SU(2^n)$ \black{in Refs.~\cite{kimura2003bloch}} as $\hat{\lambda}_i = 2^{(1-n)/2} \hat{P}_i$.
Then, the expectation value $\ev*{\hat{X}}$ of the observable $\hat{X}$ in the quantum state $\hat{\rho}$ is represented as
\begin{eqnarray}
    \ev*{\hat{X}} = \tr\qty[\hat{\rho} \hat{X}] = 2^{(n-1)/2}\vb*{\theta}\cdot\vb*{x}.\label{eq_trX}
\end{eqnarray}
Now it is clear that the estimation uncertainty of $\ev*{\hat{X}}$ can be discussed using the quantum estimation theory. By plugging the expression of $\ev*{\hat{X}}$ into the quantum Cram\'{e}r-Rao inequality~\eqref{eq_QCRB},
we find that the bound of the $\mathrm{Var}(\ev*{\hat{X}}^{\mathrm{est}})$  satisfies
\begin{eqnarray}
    \mathrm{Var}(\ev*{\hat{X}}^{\mathrm{est}}) &\geq& \frac{2^{n-1}}{N} \vb*{x}^T J(\hat{\rho}(\vb*{\theta}))^{-1} \vb*{x}\\
    &=&
    \frac{1}{N}(\Delta\hat{X})^2
\end{eqnarray}
where $J(\hat{\rho}(\vb*{\theta}))$ is the quantum Fisher information matrix of the noiseless $\hat{\rho}$
and $(\Delta\hat{X})^2$ is the variance of $\hat{X}$ in the noiseless state $\hat{\rho}$.
This equality can be achieved when we realize the optimal POVM measurement, which is
the projection measurement in the diagonal basis of the operator $\hat{X}$.

One can further analyze the variance of the unbiased estimator when the available quantum state is exposed under a noise channel $\mathcal{E}$ as~\cite{watanabe2010optimal}
\begin{eqnarray}
    \mathcal{E}(\hat{\rho}) = \frac{1}{2^n}\hat{I} + \frac{1}{2}(A\vb*{\theta} + \vb*{c})\cdot\hat{\vb*{\lambda}},
\end{eqnarray}
where $A$ is a $(2^{2n}-1)$-dimensional real matrix with its $ij$-element represented as $A_{ij} = \frac{1}{2}\mathrm{tr}[\hat{\lambda}_i\mathcal{E}(\hat{\lambda}_j)]$ and $\vb*{c} \in \mathbb{R}^{2^{2n}-1 }$ is a real vector whose $i$-th component satisfies $c_i = \frac{1}{2^n}\mathrm{tr}[\hat{\lambda}_i\mathcal{E}(\hat{I})]$.
Note that $A$ is an unital part of the Pauli transfer matrix of $\mathcal{E}$ where $\vb*{c}$ characterizes its non-unital action.
The main findings of Ref.~\cite{watanabe2010optimal} are that, when we can only prepare $N$ independent copies of a noisy state $\mathcal{E}(\hat{\rho})$, the variance $\mathrm{Var}(\ev*{\hat{X}}^{\mathrm{est}})$ of an unbiased estimator of the noiseless state obey the following inequality:
\begin{eqnarray}
    \mathrm{Var}(\ev*{\hat{X}}^{\mathrm{est}}) 
    &\geq& \frac{2^{n-1}}{N} \vb*{x}^T J(\mathcal{E}(\hat{\rho}(\vb*{\theta})))^{-1} \vb*{x} \\
    &=& \frac{1}{N}\qty(\tr\qty[\mathcal{E}(\hat{\rho})\qty((A^{-1})^T\vb*{x}\cdot \hat{\vb*{P}})^2] - \qty(\tr\qty[\mathcal{E}(\hat{\rho})\qty((A^{-1})^T\vb*{x}\cdot \hat{\vb*{P}})])^2) \\
    &=&\frac{1}{N}
    \black{(\Delta\hat{Y})^2},\label{eq_noisy_qcrb}
\end{eqnarray}
where $J(\mathcal{E}(\hat{\rho}(\vb*{\theta})))$ is now the quantum Fisher information matrix of the \black{{\it noisy}} state $\mathcal{E}(\hat{\rho})$
and $(\Delta\hat{Y})^2$ is the variance of an observable \black{$\hat{Y}$} in the {\it noisy} state $\mathcal{E}(\hat{\rho})$.
The last inequality~\eqref{eq_noisy_qcrb} tells us that the optimal estimation strategy involves an operator that satisfies $\mathcal{E}^\dagger(\hat{Y}) = \hat{X}$, whose explicit expression is given as
\begin{equation}
\hat{Y} = \qty(- \black{2^{(n-1)/2}}(A^{-1})^T\vb*{x}\cdot\vb*{c})\hat{I}+ (A^{-1})^T\vb*{x}\cdot \hat{\vb*{P}}.
\end{equation}
Note that $\mathrm{tr}[\hat{\rho} \hat{X}] = \mathrm{tr}[\mathcal{E}(\hat{\rho}) \hat{Y}]$, so $\hat{Y}$ can be interpreted as an observable that absorbs the effect of the noise.
Therefore, the way to minimize the variance of the unbiased estimator $\ev*{\hat{X}}^{\mathrm{est}}$ for the noiseless state through the measurement of the noisy state is to perform the projection measurement in the diagonal basis of the operator $\hat{Y}$.

\section{Mapping QEM methods into virtual circuit}
In this section, we explain how to map QEM methods into the virtual quantum circuit which we have presented in the main text.
The virtual circuit structure encompasses most existing quantum error mitigation methods: zero-noise extrapolation~\cite{li2017efficient,temme2017error,kandala2019error}, probabilistic error cancellation~\cite{temme2017error, endo2018practical, berg2022probabilistic}, virtual distillation~\cite{huggins_virtual_2021, koczor_exponential_2021, huo2022dualstate, czarnik2021qubit}, (generalized) quantum subspace expansion~\cite{mcclean_2017, mcclean2020decoding, yoshioka2022variational, yoshioka2022generalized}, symmetry verification/expansion~\cite{bonet-monroig2018lowcost, mcardle2019error, cai2021quantumerror}, and learning-based error mitigation~\cite{Czarnik2021errormitigation, strikis2021learning}.
For instance, we can implement zero-noise extrapolation by varying noise levels in the noise channel $\mathcal{E}_{lm}$ among each copy as in Fig.~\ref{fig_QEM_methods}(a).
We can perform probabilistic error cancellation by performing probabilistic operations using  classical registers $\hat{\rho}_{\mathrm{c},m}$ and additional operations $\mathcal{C}_{lm}$, and changing the way of post-processing according to the measurement result of $\hat{\rho}_{\mathrm{c},m}$ (Fig.~\ref{fig_QEM_methods}(b)).
We can also perform virtual distillation by performing measurements on two or more copies (Fig.~\ref{fig_QEM_methods}(c)).
We note that we can even incorporate measurement error and measurement error mitigation~\cite{kandala2017hardware, heinsoo2018rapid, bravyi2020mitigating} into the virtual circuit because noisy measurement followed by the process of measurement error mitigation can be seen as a single POVM measurement.
\begin{figure}[ht]
    \begin{center}
        \includegraphics[width=160mm]{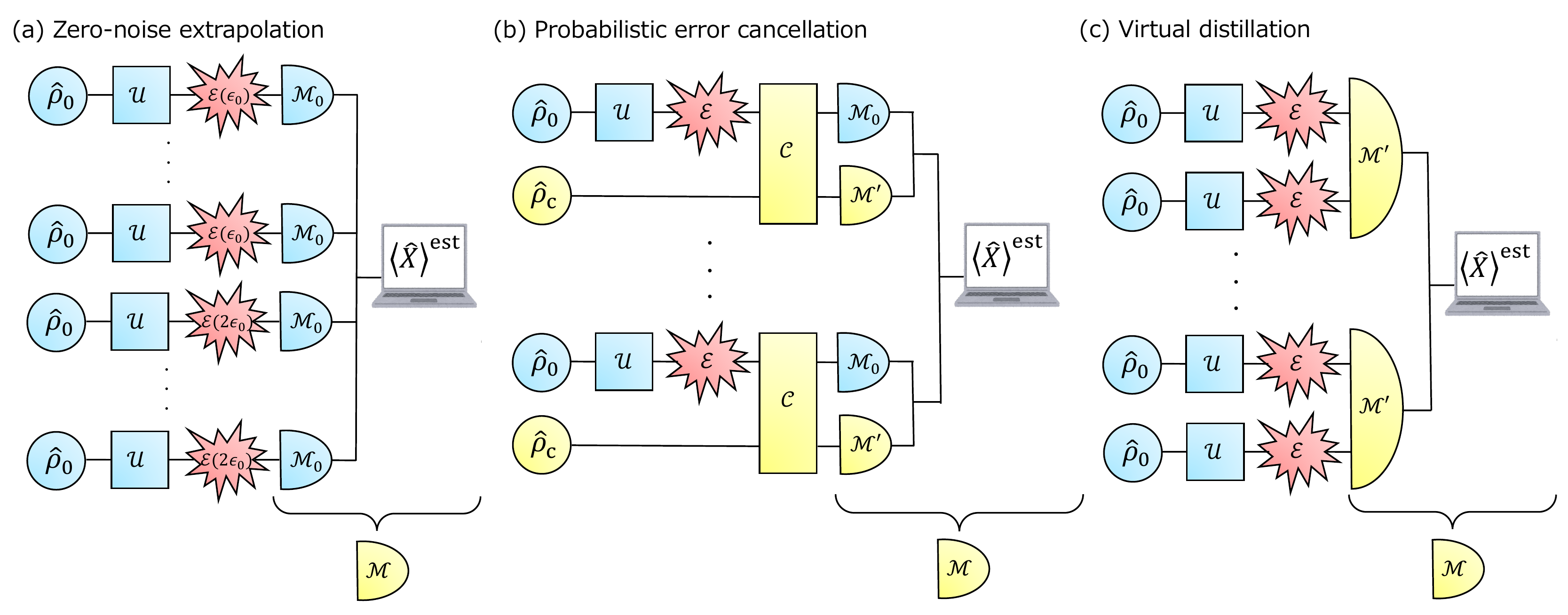}
        \caption{Ways of performing (a) zero-noise extrapolation, (b) probabilistic error cancellation, and (c) virtual distillation using the virtual quantum circuit. Measurements on each copy and post-process of those measurement results can be regarded as a single POVM measurement $\mathcal{M}$ acting on entire copies.}
        \label{fig_QEM_methods}
    \end{center}
\end{figure}

\section{Proof of Theorem \ref{thm_1} in the main text}
In this section, we present the proof of Theorem \ref{thm_1} in the main text.
The proof proceeds by progressively increasing the complexity of the virtual quantum circuit being treated. 
First, we analyze a simple case where the circuit is composed of only a single noisy gate, as expressed in Lemma \ref{lem_1}. 
Next, Lemma \ref{lem_1} is extended to Lemma \ref{lem_2}, which bounds the quantum Fisher information matrix of a noisy layered quantum circuit as represented in Fig. \ref{fig_lem_2}.
Then, as in Lemma \ref{lem_3}, we evaluate the case where we allow stochastic operations between the noisy gate.
Finally, we analyze the quantum Fisher information matrix of the virtual quantum circuit in Lemma \ref{lem_4}.
By combining Lemma \ref{lem_4} with the quantum Cram\'{e}r-Rao inequality, we can immediately prove Theorem \ref{thm_1} in the main text.
We also show $\gamma > 1$ for the unital noise satisfying the condition (II)  stated in Theorem \ref{thm_1}.

First, as in Fig. \ref{fig_lem_1}, we think of a situation where we apply a unitary gate $\mathcal{U}(\cdot) = \hat{U}\cdot \hat{U}^\dagger$ to an $n$-qubit initial state $\hat{\rho}_0$ to obtain $\hat{\rho} = \mathcal{U}(\hat{\rho}_0) = \frac{1}{2^n} \hat{I}  + \frac{1}{2}\vb*{\theta} \cdot \hat{\vb*{\lambda}}$, and measure the state to extract its information.
However, the unitary gate $\mathcal{U}$ is followed by a noise channel $\mathcal{E}$ and the POVM measurement $\mathcal{M}$ can be performed only on $\mathcal{E}\circ\mathcal{U}(\hat{\rho}_0) = \mathcal{E}(\hat{\rho}) = \frac{1}{2^n} \hat{I} + \frac{1}{2}(A\vb*{\theta} + \vb*{c})\cdot\hat{\vb*{\lambda}}$.
In such a situation, the information of $\hat{\rho}$ obtained by the POVM measurement $\mathcal{M}$ is characterized by the quantum Fisher information matrix $J(\mathcal{E}(\hat{\rho}(\vb*{\theta})))$ of the noisy state $\mathcal{E}(\hat{\rho})$, which satisfies the following lemma.

\begin{figure}[ht]
    \begin{center}
        \includegraphics[width=60mm]{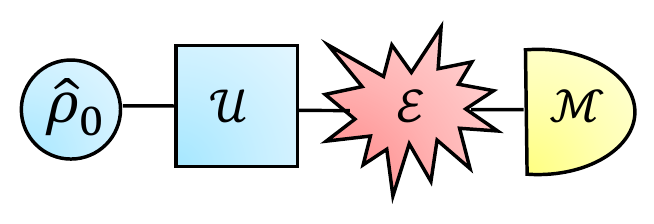}
        \caption{$\hat{\rho} = \mathcal{U}(\hat{\rho}_0)$ is prepared by applying a unitary gate $\mathcal{U}$ to the initial state $\hat{\rho}_0$, but the noise channel $\mathcal{E}$ follows the unitary gate. We can only perform POVM measurement $\mathcal{M}$ on the noisy state $\mathcal{E}(\hat{\rho})$.}
        \label{fig_lem_1}
    \end{center}
\end{figure}

\begin{lem}
    \label{lem_1}
    For the setup shown in Fig. \ref{fig_lem_1}, suppose that the noise $\mathcal{E}$ satisfies the conditions (I) and (II) stated in Theorem~\ref{thm_1} in the main text.
    Then, the quantum Fisher information matrix $J(\mathcal{E}(\hat{\rho}(\vb*{\theta})))$ satisfies
    \begin{eqnarray}
        0 < J(\mathcal{E}(\hat{\rho}(\vb*{\theta}))) \leq 2^{n-1}\beta^{-1}\gamma^{-2}I,
    \end{eqnarray}
    where $\gamma = \Gamma(\mathcal{E}) = \norm{A}^{-1}$ and $\beta$ is the largest $0 < \beta < 1$ such that $\mathcal{E}(\hat{\rho}) - \frac{\beta}{2^n}\hat{I} \geq 0$ for all $\hat{\rho}$.
\end{lem}

\begin{proof}
    From Eq. (10) of \cite{watanabe2010optimal}, the following equation holds for $\vb*{e}\in \mathbb{R}^{2^{2n}-1}$ such that $\norm{\vb*{e}} = 1$:
    \begin{eqnarray}
        \vb*{e}^TJ(\mathcal{E}(\hat{\rho}(\vb*{\theta})))^{-1}\vb*{e} = \tr\qty[\mathcal{E}(\hat{\rho})\qty((A^{-1})^T\vb*{e}\cdot \hat{\vb*{\lambda}})^2] - \qty(\tr\qty[\mathcal{E}(\hat{\rho})\qty((A^{-1})^T\vb*{e}\cdot \hat{\vb*{\lambda}})])^2.
    \end{eqnarray}
    By using the density matrix $\hat{\rho}' = \frac{1}{1-\beta}\qty(\mathcal{E}(\hat{\rho}) - \frac{\beta}{2^n}\hat{I})$, we can evaluate $\vb*{e}^TJ(\mathcal{E}(\hat{\rho}(\vb*{\theta})))^{-1}\vb*{e}$ as
    \begin{eqnarray}
        &&\vb*{e}^TJ(\mathcal{E}(\hat{\rho}(\vb*{\theta})))^{-1}\vb*{e}\\
        &=& \tr\qty[\qty(\frac{\beta}{2^n}\hat{I} + (1-\beta)\hat{\rho}')\qty((A^{-1})^T\vb*{e}\cdot \hat{\vb*{\lambda}})^2] - \qty(\tr\qty[\qty(\frac{\beta}{2^n}\hat{I} + (1-\beta)\hat{\rho}')\qty((A^{-1})^T\vb*{e}\cdot \hat{\vb*{\lambda}})])^2 \\
        &=& \beta\tr\qty[\frac{\hat{I}}{2^n}\qty((A^{-1})^T\vb*{e}\cdot \hat{\vb*{\lambda}})^2] + (1-\beta)\tr\qty[\hat{\rho}'\qty((A^{-1})^T\vb*{e}\cdot \hat{\vb*{\lambda}})^2]  - (1-\beta)^2\qty(\tr\qty[\hat{\rho}'\qty((A^{-1})^T\vb*{e}\cdot \hat{\vb*{\lambda}})])^2 \\
        &\geq& \beta\tr\qty[\frac{\hat{I}}{2^n}\qty((A^{-1})^T\vb*{e}\cdot \hat{\vb*{\lambda}})^2] + (1-\beta)^2\qty(\tr\qty[\hat{\rho}'\qty((A^{-1})^T\vb*{e}\cdot \hat{\vb*{\lambda}})^2] - \qty(\tr\qty[\hat{\rho}'\qty((A^{-1})^T\vb*{e}\cdot \hat{\vb*{\lambda}})])^2)\\
        &\geq& \beta\tr\qty[\frac{\hat{I}}{2^n}\qty((A^{-1})^T\vb*{e}\cdot \hat{\vb*{\lambda}})^2]\\
        &=& 2^{1-n}\beta\norm{(A^{-1})^T\vb*{e}}^2 \\
        &\geq& 2^{1-n} \beta\gamma^{2}.
    \end{eqnarray}
    Therefore, the inverse of the quantum \black{Fisher} information matrix satisfies
    \begin{eqnarray}
        0 < 2^{1-n} \beta\gamma^{2}I \leq J(\mathcal{E}(\hat{\rho}(\vb*{\theta})))^{-1},
    \end{eqnarray}
    and thus the quantum \black{Fisher} information matrix $J$ can be evaluated as
    \begin{eqnarray}
        0 < J(\mathcal{E}(\hat{\rho}(\vb*{\theta}))) \leq 2^{n-1}\beta^{-1}\gamma^{-2}I.
    \end{eqnarray}
\end{proof}

Next, as in Fig. \ref{fig_lem_2}, we consider the case of a layered noisy quantum circuit.
Ideally, we wish to apply a sequence of unitary gates $\qty{\mathcal{U}_l}_{l = 1}^L$ to an $n$-qubit initial state $\hat{\rho}_0$ to obtain $\hat{\rho} = \mathcal{U}_L \circ\cdots\circ\mathcal{U}_1(\hat{\rho}_0) = \frac{1}{2^n} \hat{I}  + \frac{1}{2}\vb*{\theta} \cdot \hat{\vb*{\lambda}}$, and measure the state to extract its information.
However, each unitary gate $\mathcal{U}_l$ is followed by a noise channel $\mathcal{E}_l$.
In such a situation, the state after the $L$ noisy gates is $\mathcal{E}_L\circ\mathcal{U}_L\circ\cdots\circ\mathcal{E}_1\circ\mathcal{U}_1(\hat{\rho}_0)$.
We point out that one may {\it compile} the circuit structure so that \black{it} consists only of preparation of the noiseless ideal state and a single error channel.
Namely, the noisy state can also be represented as $\mathcal{E}'(\hat{\rho}(\vb*{\theta}))$, where the effective noise channel is given as $\mathcal{E}' = \mathcal{E}_L\circ\mathcal{U}_L\circ\cdots\circ\mathcal{E}_1\circ\mathcal{U}_1\circ\mathcal{U}_1^\dagger\circ\cdots\circ\mathcal{U}_L^\dagger$.
Therefore, from Lemma \ref{lem_1}, the following lemma holds for the quantum Fisher information matrix $J(\mathcal{E}'(\hat{\rho}(\vb*{\theta})))$ of the noisy state $\mathcal{E}'(\hat{\rho})$.

\begin{figure}[ht]
    \begin{center}
        \includegraphics[width=80mm]{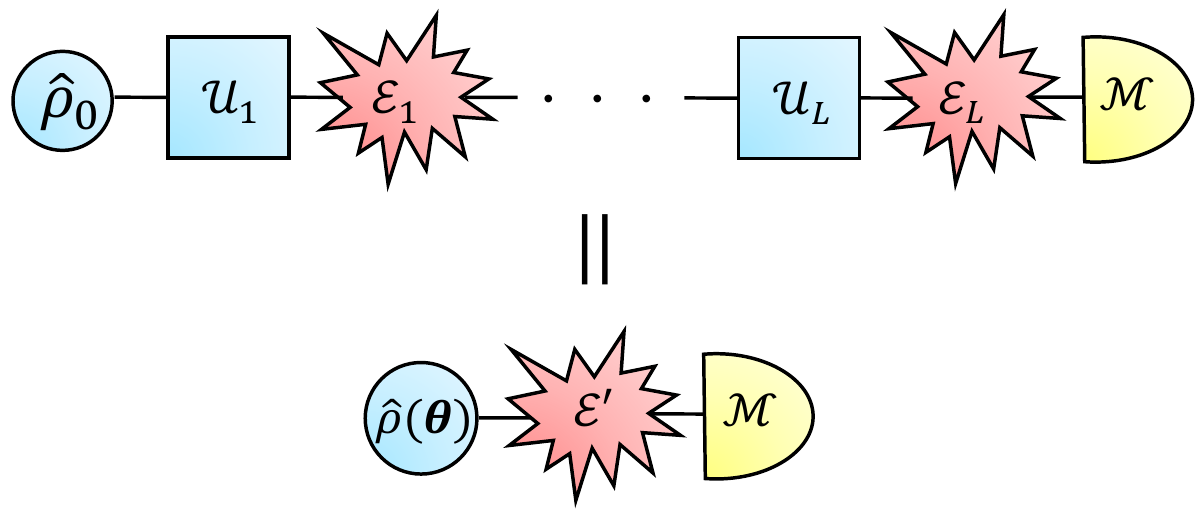}
        \caption{$\hat{\rho} = \mathcal{U}_L \circ\cdots\circ\mathcal{U}_1(\hat{\rho}_0)$ is prepared by applying a sequence of unitary gate $\qty{\mathcal{U}_l}_{l = 1}^L$ to the initial state $\hat{\rho}_0$, but the noise channel $\mathcal{E}_l$ follows the unitary gate $\mathcal{U}_l$. This situation is equivalent to the situation where preparation of the  noiseless ideal state $\hat{\rho}(\vb*{\theta})$ is followed by the noise channel $\mathcal{E} = \mathcal{E}_L\circ\mathcal{U}_L\circ\cdots\circ\mathcal{E}_1\circ\mathcal{U}_1\circ\mathcal{U}_1^\dagger\circ\cdots\circ\mathcal{U}_L^\dagger$. }
        \label{fig_lem_2}
    \end{center}
\end{figure}

\begin{lem}
    \label{lem_2}
    For the setup shown in Fig. \ref{lem_2}, suppose that the noise $\mathcal{E}_l$ satisfies the conditions (I) and (II) stated in Theorem~\ref{thm_1} in the main text.
    Then, the quantum Fisher information matrix $J(\mathcal{E}'(\hat{\rho}(\vb*{\theta})))$ satisfies
    \begin{eqnarray}
        0 < J(\mathcal{E}'(\hat{\rho}(\vb*{\theta}))) \leq 2^{n-1}\beta^{-1}\gamma^{-2L}I,
    \end{eqnarray}
    where $\gamma = \min_l\Gamma(\mathcal{E}_l)$ and $\beta$ is the largest $0 < \beta < 1$ such that $\mathcal{E}_l(\hat{\rho}) - \frac{\beta}{2^n}\hat{I} \geq 0$ for all $\hat{\rho}$ and $l$.
    Suppose further that the noise $\mathcal{E}_{l}$ is unital.
    Then, the quantum Fisher information matrix $J(\mathcal{E}'(\hat{\rho}(\vb*{\theta})))$ satisfies
    \begin{eqnarray}
        0 < J(\mathcal{E}'(\hat{\rho}(\vb*{\theta}))) \leq 2^{n-1}(1 - (1-\beta)^L)^{-1}\gamma^{-2L}I.
    \end{eqnarray}
\end{lem}

\begin{proof}
    Since $\Gamma(\mathcal{E}') \geq  \prod_{l=1}^L \Gamma(E_l) \geq \gamma^L$ and $\mathcal{E}(\hat{\rho}) - \frac{\beta}{2^n}\hat{I} \geq 0$ for any $\hat{\rho}$, the following inequality holds from Lemma \ref{lem_1}.
    \begin{eqnarray}
        0 < J(\mathcal{E}'(\hat{\rho}(\vb*{\theta}))) \leq 2^{n-1}\beta^{-1}\gamma^{-2L}I.
    \end{eqnarray}
    Also, especially when $\mathcal{E}_l$ is unital, for $0 \leq p \leq 1$ and density matrix $\hat{\rho}'$,
    \begin{eqnarray}
        \mathcal{E}_l\qty((1-p)\frac{\hat{I}}{2^n} + p\hat{\rho}') = ((1-p) + \beta p) \frac{\hat{I}}{2^n} + (1-\beta)p\hat{\rho}''
    \end{eqnarray}
    holds where $\hat{\rho}''$ is a density matrix which satisfies $\hat{\rho}'' = \frac{1}{1-\beta}\qty(\mathcal{E}_l(\hat{\rho}') - \frac{\beta}{2^n}\hat{I})$.
    Thus, by using a density matrix $\hat{\rho}'''$, $\mathcal{E}(\hat{\rho})$ can be expressed as
    \begin{eqnarray}
        \mathcal{E}(\hat{\rho}) = (1 - (1-\beta)^L)\frac{\hat{I}}{2^n} + (1-\beta)^L\hat{\rho}'''.
    \end{eqnarray}
    Therefore, since $\mathcal{E}(\hat{\rho}) - \frac{(1 - (1-\beta)^L)}{2^n}I \geq 0$ for all $\hat{\rho}$,
    \begin{eqnarray}
        0 < J(\mathcal{E}'(\hat{\rho}(\vb*{\theta}))) \leq 2^{n-1}(1 - (1-\beta)^L)^{-1}\gamma^{-2L}I
    \end{eqnarray}
    holds from Lemma \ref{lem_1}.
\end{proof}

Then, as shown in Fig. \ref{fig_lem_3}, we allow stochastic operations by adding classical register $\hat{\rho}_\mathrm{c}$ coupled with the system qubits via the additional operation $\mathcal{C}_l$.
Classical register ${\hat\rho}_{\mathrm{c}}$ is initialized with probabilistic mixtures of computational bases as $\hat{\rho}_{\mathrm{c}} = \sum_{i}p_{i}\ketbra{i}$, and additional operation $\mathcal{C}_{l}$ performs unitary operation $\mathcal{C}_{li}$ according to the state of the classical register as $\mathcal{C}_{l} = \sum_{i} \mathcal{C}_{li}\otimes\ketbra{i}$.

In this case, the quantum state before the measurement is expressed as $\mathcal{E}'(\hat{\rho}(\vb*{\theta})\otimes\hat{\rho}_\mathrm{c}) = \sum_{i} p_{i}\mathcal{E}'_{i}(\hat{\rho}(\vb*{\theta})) \otimes \ketbra{i}$, where $\mathcal{E}' = \mathcal{E}_{L}\circ\mathcal{U}_L\circ\mathcal{C}_{L}\circ\cdots\circ\mathcal{E}_{1}\circ\mathcal{U}_1\circ\mathcal{C}_{1}\circ\mathcal{U}_1^{-1}\circ\cdots\circ\mathcal{U}_L^{-1}$ and $\mathcal{E}_{i} = \mathcal{E}_L\circ\mathcal{U}_L\circ\mathcal{C}_{L-1i}\circ\cdots\circ\mathcal{E}_1\circ\mathcal{U}_1\circ\mathcal{C}_{0i}\circ\mathcal{U}_1^\dagger\circ\cdots\circ\mathcal{U}_L^\dagger$ represents the effective noise channel.
Therefore, from the convexity and additivity of the quantum Fisher information matrix:
\begin{eqnarray}
    J(\mathcal{E}'(\hat{\rho}(\vb*{\theta})\otimes\hat{\rho}_\mathrm{c})) 
    &\leq& \sum_{i} p_{i} J(\mathcal{E}'_{i}(\hat{\rho}(\vb*{\theta})) \otimes \ketbra{i}) \\
    &=& \sum_{i} p_{i} J(\mathcal{E}'_{i}(\hat{\rho}(\vb*{\theta}))),
\end{eqnarray}
the following lemma can be derived in the same way as in Lemma \ref{lem_2}.

\begin{figure}[ht]
    \begin{center}
        \includegraphics[width=100mm]{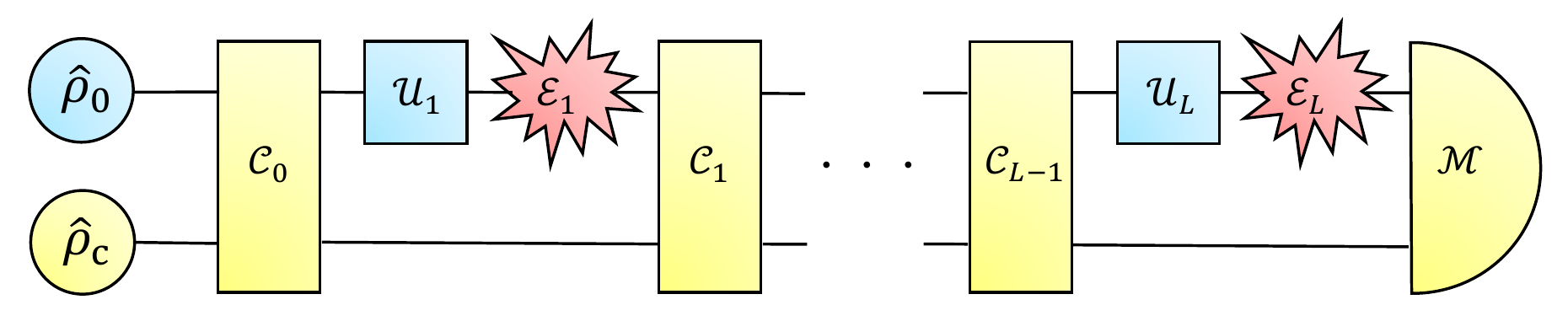}
        \caption{The setup where we allow stochastic operations by adding classical register $\hat{\rho}_\mathrm{c}$ coupled with the system qubits via the additional operation $\mathcal{C}_i$.}
        \label{fig_lem_3}
    \end{center}
\end{figure}

\begin{lem}
    \label{lem_3}
    For the setup shown in Fig. \ref{fig_lem_3}, suppose that the noise $\mathcal{E}_l$ satisfies the conditions (I) and (II) stated in Theorem~\ref{thm_1} in the main text.
    Then, the quantum Fisher information matrix $J(\mathcal{E}'(\hat{\rho}(\vb*{\theta})\otimes\hat{\rho}_\mathrm{c}))$ satisfies
    \begin{eqnarray}
        0 < J(\mathcal{E}'(\hat{\rho}(\vb*{\theta})\otimes\hat{\rho}_\mathrm{c})) \leq 2^{n-1}\beta^{-1}\gamma^{-2L}I,
    \end{eqnarray}
    where $\gamma = \min_l \Gamma(\mathcal{E}_l)$ and $\beta$ is the largest $0 < \beta < 1$ such that $\mathcal{E}_l(\hat{\rho}) - \frac{\beta}{2^n}\hat{I} \geq 0$ for all $\hat{\rho}$ and $l$.
    Suppose further that the noise $\mathcal{E}_{l}$ is unital.
    Then, the quantum Fisher information matrix $J(\mathcal{E}'(\hat{\rho}(\vb*{\theta})\otimes\hat{\rho}_\mathrm{c}))$ satisfies
    \begin{eqnarray}
        0 < J(\mathcal{E}'(\hat{\rho}(\vb*{\theta})\otimes\hat{\rho}_\mathrm{c})) \leq 2^{n-1}(1 - (1-\beta)^L)^{-1}\gamma^{-2L}I.
    \end{eqnarray}
\end{lem}

Finally, we consider the virtual quantum circuit in Fig. \ref{fig_mitigation_circuit} in the main text, where the POVM measurement $\mathcal{M}$ can be performed on $N$ copies of the quantum circuit in Fig. \ref{fig_lem_3} with varying $\mathcal{E}_l$ and $\mathcal{C}_l$ among each copy.
Then, from the additivity of the quantum Fisher information matrix:
\begin{eqnarray}
    J\qty(\bigotimes_{m=1}^N\mathcal{E}'_{m}(\hat{\rho}(\vb*{\theta})\otimes\hat{\rho}_{\mathrm{c},m})) = \sum_{m=1}^N J(\mathcal{E}'_{m}(\hat{\rho}(\vb*{\theta})\otimes\hat{\rho}_{\mathrm{c},m}))
\end{eqnarray}
and Lemma \ref{lem_3}, we obtain the following lemma.

\begin{lem}
    \label{lem_4}
    Suppose that the noise channel $\mathcal{E}_{lm}$ satisfies the conditions (I) and (II) stated in Theorem~\ref{thm_1} in the main text.
    Then, the quantum Fisher information matrix $J$ of the state before the measurement satisfies
    \begin{eqnarray}
        0 < J\qty(\bigotimes_{m=1}^N\mathcal{E}'_{m}(\hat{\rho}(\vb*{\theta})\otimes\hat{\rho}_{\mathrm{c},m})) \leq 2^{n-1}N\beta^{-1}\gamma^{-2L}I,
    \end{eqnarray}
    where 
    $\beta$ is the largest $0 < \beta < 1$ such that $\mathcal{E}_{lm}(\hat{\rho}) - \frac{\beta}{2^n}\hat{I} \geq 0$ for all $\hat{\rho}$, $l$, and $m$.
    Suppose further that the noise $\mathcal{E}_{lm}$ is unital.
    Then, the quantum Fisher information matrix $J$ satisfies
    \begin{eqnarray}
        0 < J\qty(\bigotimes_{m=1}^N\mathcal{E}'_{m}(\hat{\rho}(\vb*{\theta})\otimes\hat{\rho}_{\mathrm{c},m})) \leq 2^{n-1}N(1 - (1-\beta)^L)^{-1}\gamma^{-2L}I.
    \end{eqnarray}
\end{lem}
By combining Lemma \ref{lem_4} with the quantum Cram\'{e}r-Rao inequality, we immediately obtain Theorem \ref{thm_1} in the main text.

We here show $\gamma > 1$ for unital noise channels satisfying the condition (II) of the Theorem \ref{thm_1} in the main text.
The proof proceeds in the same way as in Corollary A.1 of Refs.~\cite{sagawa2022entropy}.
Let $\hat{\rho} =  \frac{1}{2^n} \hat{I}  + \frac{1}{2}\vb*{\theta} \cdot \hat{\vb*{\lambda}}\neq \frac{1}{2^n} \hat{I}$ be a quantum state and $\rho = \sum_i p_i\ketbra{i}$ be its spectral decomposition.
We consider the matrix
\begin{eqnarray}
    \begin{pmatrix}
        \mathcal{E}_{lm}(\hat{\rho}^2) & \mathcal{E}_{lm}(\hat{\rho}) \\
        \mathcal{E}_{lm}(\hat{\rho}) & \mathcal{E}_{lm}(\hat{I}) \\
    \end{pmatrix}
    =
    \sum_i
    \begin{pmatrix}
        p_i^2 & p_i \\
        p_i & 1 \\
    \end{pmatrix}
    \otimes \mathcal{E}_{lm}(\ketbra{i}) > 0.
\end{eqnarray}
The last inequality holds because $\mathcal{E}_{lm}(\ketbra{i}) > 0$ from condition (II) and there exist $i\neq i'$ such that $p_i \neq p_{i'}$.
This implies 
\begin{eqnarray}
    &&
    \begin{pmatrix}
        \hat{I} & -\mathcal{E}_{lm}(\hat{\rho})\mathcal{E}_{lm}(\hat{I})^{-1} \\
        0 & \hat{I} \\
    \end{pmatrix}
    \begin{pmatrix}
        \mathcal{E}_{lm}(\hat{\rho}^2) & \mathcal{E}_{lm}(\hat{\rho}) \\
        \mathcal{E}_{lm}(\hat{\rho}) & \mathcal{E}_{lm}(\hat{I}) \\
    \end{pmatrix}
    \begin{pmatrix}
        \hat{I} & -\mathcal{E}_{lm}(\hat{\rho})\mathcal{E}_{lm}(\hat{I})^{-1} \\
        0 & \hat{I} \\
    \end{pmatrix}^\dagger\\
    &=&
    \begin{pmatrix}
        \mathcal{E}_{lm}(\hat{\rho}^2) - \mathcal{E}_{lm}(\hat{\rho}^2)\mathcal{E}_{lm}(\hat{I})^{-1}\mathcal{E}_{lm}(\hat{\rho}) & 0 \\
        0 & \mathcal{E}_{lm}(\hat{I}). \\
    \end{pmatrix}\\
    &>&0.
\end{eqnarray}
Thus, 
\begin{eqnarray}
    \mathrm{\tr}[\hat{\rho}^2] 
    =\mathrm{\tr}[\mathcal{E}_{lm}(\hat{\rho}^2)] 
    > \mathrm{\tr}[\mathcal{E}_{lm}(\hat{\rho})\mathcal{E}_{lm}(\hat{I})^{-1}\mathcal{E}_{lm}(\hat{\rho})] 
    = \mathrm{\tr}[\mathcal{E}_{lm}(\hat{\rho})^2].
\end{eqnarray}
This means
\begin{eqnarray}
    \norm{\vb*{\theta}} > \norm{A_{lm}\vb*{\theta}}
\end{eqnarray}
for $\vb*{\theta}\neq \vb*{0}$.
Therefore, $\Gamma(\mathcal{E}_{lm}) > 1$ for unital noise satisfying the condition (II), meaning that $\gamma > 1$.
We note that this is not the case for some non-unital noise and $\Gamma(\mathcal{E}_{lm})$ may even be less than 1~\cite{ozawa2000entanglement}.

\section{Haar integration and unitary $t$-design}\label{app:haar_integration}
Before presenting the details of Theorem \ref{thm_2} in the main text, here we briefly review the formulas regarding the Haar integrals, namely the average over unitary group with respect to the Haar measure.
This problem was initially considered by Weingarten in 1978 in the asymptotic limit of large matrix size~\cite{weingarten1978asymptotic} and then expanded to expressions for finite-size matrix by Refs.~\cite{collins2003moments, collins2006integration}.
The element-wise formula involving $p$-th power of random unitary matrix elements is explicitly given as
\begin{eqnarray}
     \mathbb{E}_{U} \left[ U_{i_1j_1} \cdots U_{i_p j_p} U^*_{k_1 l_1} \cdots U^*_{k_p l_p} \right] = \sum_{\sigma, \tau \in \mathcal{S}_p}
     \delta_{i_1 k_{\sigma(1)}} \cdots \delta_{i_p k_{\sigma(p)}} \delta_{j_1 l_{\tau(1)}} \cdots \delta_{j_p l_{\tau(p)}} {\rm Wg}_{p,d}(\sigma^{-1}\tau),
\end{eqnarray}
where \red{$\mathbb{E}_U$ denotes the average under Haar random unitary and} $\sigma, \tau$ is a permutation over $p$ indices taken from the entire permutation group $\mathcal{S}_p$. 
This implies that there are in general $(p!)^2$ terms in total for the $p$-th moment formula.
The coefficient $\rm Wg$ is the so-called Weingarten function~\cite{weingarten1978asymptotic} which is given for lower moments as
\begin{eqnarray}
    (p=1)\  &:&\  {\rm Wg}_{1,d}([1]) = \frac{1}{d}, \\
    (p=2) \ &:&\ {\rm Wg}_{2,d}([1,1]) = \frac{1}{d^2 - 1}, \ \ 
    {\rm Wg}_{2,d}([2]) = -\frac{1}{d(d^2 - 1)}.
\end{eqnarray}

\noindent{\bf First moment.} 
The element-wise formula is explicitly given as
\begin{eqnarray}
    \mathbb{E}_{U} \left[ U_{ij} U^*_{kl} \right] = \frac{1}{d} \delta_{ik}\delta_{jl}.
\end{eqnarray}
From this equation, we derive the well-known formula as
\begin{eqnarray}
    \mathbb{E}_{U} \left[(UAU^\dag)_{ij} \right] &=& \sum_{kl} \mathbb{E}_U\qty[U_{ik}A_{kl}U^*_{jl}] \\
    &=& \frac{1}{d}\tr[A]\delta_{ij}.
\end{eqnarray}

\noindent {\bf Second moment.}
Similar to the first moment, we can also write down the formula as
\begin{eqnarray}
    \mathbb{E}_{U} \left[ U_{i_1 j_1}U_{i_2 j_2} U^*_{k_1 l_1} U^*_{k_2 l_2} \right] = \frac{
    \delta_{i_1 k_1}\delta_{i_2 k_2} \delta_{j_1 l_1} \delta_{j_2 l_2}
    + \delta_{i_1 k_2}\delta_{i_2 k_1} \delta_{j_1 l_2} \delta_{j_2 l_1}
    }{d^2 - 1} \\
    - \frac{
    \delta_{i_1 k_2} \delta_{i_2 k_1} \delta_{j_1 l_1} \delta_{j_2 l_2}
    + \delta_{i_1 k_1} \delta_{i_2 k_2} \delta_{j_1 l_2} \delta_{j_2 l_1}
    }{d(d^2 - 1)}.
\end{eqnarray}
This yields the second-moment Haar integral formula as
\begin{eqnarray}
    \mathbb{E}_{U} \left[ \tr[UAU^\dag B U C U^\dag D] \right] &=& \frac{\tr[A C]\tr[B]\tr[D] + \tr[A]\tr[C] \tr[BD]}{d^2 -1} - \frac{\tr[AC]\tr[BD] + \tr[A]\tr[B] \tr[C] \tr[D]}{d(d^2 - 1)},\nonumber \\
    \mathbb{E}_{U} \left[ \tr[UAU^\dag B] \tr[U C U^\dag D] \right]&=& \frac{\tr[AC]\tr[BD] + \tr[A]\tr[B] \tr[C] \tr[D]}{d^2 -1} - \frac{\tr[A C]\tr[B]\tr[D] + \tr[A]\tr[C] \tr[BD]}{d(d^2 - 1)}.\nonumber
\end{eqnarray}

We can define unitary $t$-design~\cite{dankert2009exact} by using $t$-th moment of Haar integration: unitary $t$-design is a set of unitaries $\qty{U_1,\cdots,U_K}$ which satisfies
\begin{eqnarray}
    \frac{1}{K}\sum_{k=1}^K \left[ (U_k)_{i_1j_1} \cdots (U_k)_{i_t j_t} (U_k)^*_{k_1 l_1} \cdots (U_k)^*_{k_t l_t} \right] 
    = \mathbb{E}_{U} \left[ U_{i_1j_1} \cdots U_{i_t j_t} U^*_{k_1 l_1} \cdots U^*_{k_t l_t} \right].
\end{eqnarray}
For instance, $n$-qubit Clifford group is a unitary 3-design~\cite{webb2015clifford, zhu2017multiqubit} (thus it is also a unitary 2-design).
In the following section, we only think of the second moment of Haar integration at most, so we denote the average under unitary $2$-design as $\mathbb{E}_U$.

\section{Details and Proof of Theorem 2 in the main text}
In this section, we present details and proof of Theorem \ref{thm_2} in the main text.
Let us first reformulate the problem setup.
We think of running a noisy random circuit exposed to local noise, where each unitary gate is chosen randomly from unitary 2-design, and obtain an unbiased estimator of an expectation value of a traceless observable $\hat{X}$ for a noiseless circuit.
Concretely, We are given $N$ copies of the noisy state $\hat{\rho}_{\rm noisy} = \mathcal{E}_L\circ\mathcal{U}_L\circ\cdots\circ\mathcal{E}_1\circ\mathcal{U}_1(\hat{\rho}_0)$, where $\mathcal{E}_l = \bigotimes_{i=1}^n \mathcal{E}_l^{(i)}$ is a local noise and $\mathcal{U}_l(\cdot) = U_l\cdot U_l^\dagger$ is drawn from unitary 2-design, and we aim to obtain an unbiased estimator of the expectation value $\ev*{\hat{X}} \equiv \mathrm{tr}[\hat{\rho}\hat{X}]$, where $\hat{\rho} = \mathcal{U}_L\circ\cdots\circ\mathcal{U}_1(\hat{\rho}_0)$.
We also assume that the initial state $\hat{\rho}_0$ is pure and the noise channel $\mathcal{E}_l$ satisfies the condition (I) in Theorem \ref{thm_1}.

For $2^n$-dimensional CPTP map $\mathcal{E}$ satisfying the condition (I), we can define its inverse map $\mathcal{E}^{-1}(\cdot) = \sum_i p_i E_i \cdot E_i^\dagger = \sum_i q_{ij} P_i\cdot P_j$, where $p_i\in\mathbb{R}$, $q_{ij}\in\mathbb{C}$ satisfying $q_{ij} = q_{ji}^*$, and $P_i$ are $n$-qubit Pauli operators including the identity operator~\cite{jiang2021physical}.
Since $\mathcal{E}^{-1}$ is also trace-preserving~\cite{jiang2021physical}, we can show $\sum_ip_iE_i^\dagger E_i =  \sum_i q_{ij} P_jP_i = \hat{I}$.
Especially when $\mathcal{E}$ is unital, we can also show $\sum_ip_iE_iE_i^\dagger =  \sum_i q_{ij} P_iP_j = \hat{I}$.
By using the inverse $\mathcal{E}^{-1}$, we can define the noise strength of $\mathcal{E}$ as 
\begin{eqnarray}
    \nu(\mathcal{E}^{-1}) = \mathrm{tr}[(\red{\hat{I}}\otimes\mathcal{E}^{-1}(\ketbra{\Gamma}))^2]/d^2 = \sum_{ij}p_ip_j \abs*{\mathrm{tr}[E_iE_j^\dagger]}^2/d^2 = \sum_{ij}\abs{q_{ij}}^2,
\end{eqnarray}
where $\ket{\Gamma} = \sum_{i=1}^d\ket{i}\ket{i}$ is the maximally entangled state and $d=2^n$ is the dimension of the system.
From definition, $\nu(\mathcal{E}^{-1})$ is multiplicative: for local noise channel $\mathcal{E} = \bigotimes_{i=1}^n \mathcal{E}^{(i)}$, 
\begin{eqnarray}
    \nu(\mathcal{E}^{-1}) = \prod_{i=1}^n \nu((\mathcal{E}^{(i)})^{-1}).
\end{eqnarray}

Now we are ready to state the details of Theorem \ref{thm_2}, which gives the scaling of the lower bound of $N$.
We first present the result for the unital noise.
\begin{thm}
    \label{thm_3}
    Let $U_1, U_2, ..., U_{L-1}, U_L$ be $n$-qubit unitary gate drawn from a set of random unitary that form unitary 2-design, and $\mathcal{E}_l = \bigotimes_{i=1}^N \mathcal{E}_l^{(i)} $ be  a local unital noise channel \red{satisfying condition (I) in the Theorem \ref{thm_1}}.
    Then, \red{the average of the number of copies of $\rho_{\mathrm{noisy}}$} required to perform unbiased estimation of $\ev*{\hat{X}}$ with standard deviation $\varepsilon$ over $\{U_1, ..., U_L\}$ in the above setup is lower bounded as
    \begin{eqnarray}
        \mathbb{E}_{U_1\cdots U_L}[N]
    &\geq& \frac{\norm{\vb*{x}}^2}{\varepsilon^2}\qty(\prod_{L=1}^{l-1} \qty(\frac{4^n\prod_{i=1}^n\nu((\mathcal{E}_{l}^{(i)})^{-1})-1}{4^n-1}) - \frac{2^n-2}{4^n-1})\\
    &\sim& \frac{\norm{\vb*{x}}^2}{\varepsilon^2}\prod_{l=1}^{L-1}\prod_{i=1}^n \nu((\mathcal{E}_{l}^{(i)})^{-1}),
    \end{eqnarray}
    implying the exponential growth with both qubit count n and depth L.
\end{thm}

\begin{proof}
    Let us define $\mathcal{N}_l = \mathcal{E}_l\circ\mathcal{U}_l\circ\cdots\circ\mathcal{E}_1\circ\mathcal{U}_1\circ\mathcal{U}_i^\dagger\circ\cdots\circ\mathcal{U}_l^\dagger$.
    From the monotonicity of quantum Fisher information matrix and Eq. (10) of Refs.~\cite{watanabe2010optimal}, we can lower bound the average of the cost $N$ as
    \begin{eqnarray}
        \mathbb{E}_{U_1\cdots U_L}[N]
        \geq \frac{1}{\varepsilon^2}\mathbb{E}_{U_1\cdots U_L}\qty[\mathrm{tr}[\mathcal{E}_{L-1}\circ\mathcal{U}_{L-1}\circ\cdots\circ\mathcal{E}_1\circ\mathcal{U}_1(\hat{\rho}_0)((\mathcal{N}^{\dagger}_{L-1})^{-1}\circ\mathcal{U}_L^\dag(\hat{X}))^2] - \mathrm{tr}[\hat{\rho} \hat{X}]^2].
    \end{eqnarray}
    
    We first evaluate the second term.
    By defining $\hat{\rho}' = \mathcal{U}_{L-1}\circ\cdots\circ\mathcal{U}_1(\hat{\rho}_0)$ and using the second-moment formula of the Haar integral for $U_L$, we can compute it explicitly as
    \begin{eqnarray}
        \mathbb{E}_{U_1\cdots U_L}\qty[\mathrm{tr}[\mathcal{U}_L\circ\cdots\circ\mathcal{U}_1(\hat{\rho}_0) \hat{X}]^2]
        &=&
        \mathbb{E}_{U_1\cdots U_L} \left[
        \mathrm{tr}[U_L \hat{\rho}' U_L^\dag \hat{X}]^2
        \right]
        \\
        &=&
        \mathbb{E}_{U_1\cdots U_{L-1}}\left[
        \left(
        \frac{\mathrm{tr}[\hat{\rho}'^2]}{d^2 - 1}
        -\frac{\mathrm{tr}[\hat{\rho}']^2}{d(d^2-1)}
        \right)\mathrm{tr}[\hat{X}^2]
        \right]\\
        &=& \frac{1}{d(d+1)}\mathrm{tr}[\hat{X}^2],
    \end{eqnarray}
    where we used $\mathrm{tr}[X]=0$ in the second equation and $\tr[\rho'^2] = 1$ for pure state $\hat{\rho}'$ to obtain the last equation.

    To evaluate the first term, we first make use of the first-moment formula of the Haar integral for $U_1$ as
    \begin{eqnarray}
        &&\mathbb{E}_{U_1\cdots U_L}\qty[\mathrm{tr}[\mathcal{E}_{L-1}\circ\mathcal{U}_{L-1}\circ\cdots\circ\mathcal{E}_1\circ\mathcal{U}_1(\hat{\rho}_0)((\mathcal{N}^{\dagger}_{L-1})^{-1}\circ\mathcal{U}_L^\dag(\hat{X}))^2]] \\
        &=& \mathbb{E}_{U_2\cdots U_L}\qty[\mathrm{tr}[\mathcal{E}_{L-1}\circ\mathcal{U}_{L-1}\circ\cdots\circ\mathcal{E}_1(\mathbb{E}_{U_1}[\mathcal{U}_1(\hat{\rho}_0)])((\mathcal{N}^{\dagger}_{L-1})^{-1}\circ\mathcal{U}_L^\dag(\hat{X}))^2]]\\
        &=&\mathbb{E}_{U_2\cdots U_L}\qty[\mathrm{tr}[\mathcal{E}_{L-1}\circ\mathcal{U}_{L-1}\circ\cdots\circ\mathcal{E}_1(\hat{I}/d)((\mathcal{N}^{\dagger}_{L-1})^{-1}\circ\mathcal{U}_L^\dag(\hat{X}))^2]]\\
        &=& \frac{1}{d}\mathbb{E}_{U_2\cdots U_L}\qty[\mathrm{tr}[((\mathcal{N}^{\dagger}_{L-1})^{-1}\circ\mathcal{U}_L^\dag(\hat{X}))^2]],
    \end{eqnarray}
    where we used the unitality of the noise channel $\mathcal{E}_l$ in the third equation.
    Next, by expanding $\mathcal{N}_{L-1}^{-1}$ as $\mathcal{N}_{L-1}^{-1}(\cdot) = \sum_i s_iN_i \cdot N_i^\dag$, we can use the second-moment formula of the Haar integral for $U_{L}$ as
    \begin{eqnarray}
        \frac{1}{d}\mathbb{E}_{U_L}\qty[\mathrm{tr}[((\mathcal{N}^{\dagger}_{L-1})^{-1}\circ\mathcal{U}_L^\dag(\hat{X}))^2]]
        &=& \frac{1}{d}\sum_{ij}s_is_j\mathbb{E}_{U_L}\qty[\mathrm{tr}[(N_i^\dag U_L^\dag\hat{X}U_LN_i)(N_j^\dag U_L^\dag\hat{X}U_LN_j)]] \\
        &=& \frac{1}{d}\sum_{ij}s_is_j\mathbb{E}_{U_L}\qty[\mathrm{tr}[(U_L^\dag\hat{X}U_LN_iN_j^\dag U_L^\dag\hat{X}U_LN_jN_i^\dag ]] \\
        &=&\frac{1}{d}\sum_{ij}s_is_j \qty(\frac{\mathrm{tr}[N_iN_j^\dag]\mathrm{tr}[N_jN_i^\dag]\mathrm{tr}[\hat{X}^2]}{d^2-1} - \frac{\mathrm{tr}[N_iN_j^\dag N_jN_i^\dag]\mathrm{tr}[\hat{X}^2]}{d(d^2-1)})\\
        &=& \frac{d}{d^2-1} \qty(\nu(\mathcal{N}^{-1}_{L-1}) - \frac{1}{d^2})\mathrm{tr}[\hat{X}^2],
    \end{eqnarray}
    where we used $\mathrm{tr}[X]=0$ in the third equation.
    Therefore, the average of the cost $N$ can be lower bounded as
    \begin{eqnarray}
        \mathbb{E}_{U_1\cdots U_L}[N] 
        \geq  \frac{1}{\varepsilon^2}\frac{d}{d^2-1}\mathrm{tr}[\hat{X}^2] \qty(\mathbb{E}_{U_2\cdots U_{L-1}}[\nu(\mathcal{N}^{-1}_{L-1})] - \frac{1}{d^2} - \frac{d-1}{d^2}).
        \label{eq_Nav}
    \end{eqnarray}
    
Next, we evaluate $\mathbb{E}_{U_l}[\nu(\mathcal{N}_l^{-1})]$.
Let us assume that $\mathcal{E}_{l}^{-1}$ and $\mathcal{N}_{l-1}^{-1}$ can be expanded as $\mathcal{E}_{l}^{-1}(\cdot) = \sum_ip_iE_i\cdot E_i^\dag$ and $\mathcal{N}_{l-1}^{-1}(\cdot) = \sum_i s_iN_i \cdot N_i^\dag$.
Then, since $\mathcal{N}_l^{-1} = \mathcal{U}_l\circ\mathcal{N}_{l-1}^{-1}\circ\mathcal{U}_l^\dag\circ\mathcal{E}_l^{-1}$, we can expand $\mathcal{N}_l^{-1}$ as $\mathcal{N}_l^{-1}(\cdot) = \sum_{ij} p_i s_j U_l N_j U_l^\dag E_i \cdot E_i^\dag U_l N_j^\dag U_l^\dag$.
Thus, we can evaluate $\mathbb{E}_{U_l}[\nu(\mathcal{N}_l^{-1})]$ as
\begin{eqnarray}
    \mathbb{E}_{U_l}[\nu(\mathcal{N}_l^{-1})]
    &=& \sum_{ijkl}p_ip_ks_js_l\mathbb{E}_{U_l}\qty[\mathrm{tr}[U_l N_j U_l^\dag E_i E_k^\dag U_l N_l^\dag U_l^\dag]\mathrm{tr}[U_l N_l U_l^\dag E_k E_i^\dag U_l N_j^\dag U_l^\dag]]/d^2 \\
    &=& \sum_{ijkl}p_ip_ks_js_l\mathbb{E}_{U_l}\qty[\mathrm{tr}[U_l N_l^\dag N_j U_l^\dag E_i E_k^\dag]\mathrm{tr}[U_l N_j^\dag N_l U_l^\dag E_k E_i^\dag ]]/d^2 \\
    &=& \sum_{ijkl}p_ip_ks_js_l \biggl(\frac{\mathrm{tr}[N_l^\dag N_jN_j^\dag N_l]\mathrm{tr}[E_i E_k^\dag E_k E_i^\dag] + \mathrm{tr}[N_l^\dag N_j]\mathrm{tr}[N_j^\dag N_l]\mathrm{tr}[E_i E_k^\dag] \mathrm{tr}[E_k E_i^\dag]}{d^2(d^2 -1)} \\
    &&- \frac{\mathrm{tr}[N_l^\dag N_jN_j^\dag N_l]\mathrm{tr}[E_i E_k^\dag] \mathrm{tr}[E_k E_i^\dag] + \mathrm{tr}[N_l^\dag N_j]\mathrm{tr}[N_j^\dag N_l]\mathrm{tr}[E_i E_k^\dag E_k E_i^\dag]}{d^3(d^2 - 1)}\biggr)\\
    &=& \frac{d^2+d^4\nu(\mathcal{N}_{l-1}^{-1})\nu(\mathcal{E}_{l}^{-1})}{d^2(d^2-1)} - \frac{d^3\nu(\mathcal{E}_{l}^{-1}) + d^3\nu(\mathcal{N}_{l-1}^{-1})}{d^3(d^2-1)}\\
    &=& \frac{d^2\nu(\mathcal{E}_{l}^{-1})-1}{d^2-1}\qty(\nu(\mathcal{N}_{l-1}^{-1})-\frac{1}{d^2}) + \frac{1}{d^2}.
    \label{eq_nu}
\end{eqnarray}

By combining Eq. (\ref{eq_Nav}) and Eq. (\ref{eq_nu}), we obtain the following inequality:
\begin{eqnarray}
    \mathbb{E}_{U_1\cdots U_L}[N]
    &\geq& \frac{1}{\varepsilon^2}\frac{d}{d^2-1}\mathrm{tr}[\hat{X}^2] \qty(\frac{d^2-1}{d^2}\prod_{l=1}^{L-1} \qty(\frac{d^2\nu(\mathcal{E}_{l}^{-1})-1}{d^2-1}) - \frac{d-2}{d^2}) \\
    &=& \frac{\norm{\vb*{x}}^2}{\varepsilon^2}\qty(\prod_{l=1}^{L-1} \qty(\frac{d^2\nu(\mathcal{E}_{l}^{-1})-1}{d^2-1}) - \frac{d-2}{d^2-1})\\
    &=& \frac{\norm{\vb*{x}}^2}{\varepsilon^2}\qty(\prod_{l=1}^{L-1} \qty(\frac{4^n\prod_{i=1}^n\nu((\mathcal{E}_{l}^{(i)})^{-1})-1}{4^n-1}) - \frac{2^n-2}{4^n-1})\\
    &\sim& \frac{\norm{\vb*{x}}^2}{\varepsilon^2}\prod_{l=1}^{L-1}\prod_{i=1}^n \nu((\mathcal{E}_{l}^{(i)})^{-1}).
\end{eqnarray}
\end{proof}

Next, let us discuss the non-unital case.
For inverse map $\mathcal{E}^{-1}(\cdot) = \sum_i p_i E_i \cdot E_i^\dagger$, we define another quantity characterizing the noise channel as 
\begin{eqnarray}
    \eta(\mathcal{E}^{-1}) = \mathrm{tr}[(\mathcal{E}^{-1}(\hat{I}))^2] /d^2 = \sum_{ij}p_ip_j\mathrm{tr}[E_iE_i^\dag E_jE_j^\dag] /d^2.
\end{eqnarray}
From definition, $\eta(\mathcal{E}^{-1})$ is also multiplicative: for local noise channel $\mathcal{E} = \bigotimes_{i=1}^n \mathcal{E}^{(i)}$, 
\begin{eqnarray}
    \eta(\mathcal{E}^{-1}) = \prod_{i=1}^n \eta((\mathcal{E}^{(i)})^{-1}).
\end{eqnarray}
For simplicity, let us assume the noise to be homogeneous such shat  $\mathcal{E}_l = \mathcal{E} = \bigotimes_{i=1}^n \mathcal{E}^{(0)}$ for all $l=1,\cdots,L$.
Then, we obtain the following theorem for non-unital noise.
\begin{thm}
    \label{thm_4}
    Let $U_1, U_2, ..., U_{L-1}, U_L$ be $n$-qubit unitary gate drawn from a set of random unitary that form unitary 2-design, and $\mathcal{E} = \bigotimes_{i=1}^N \mathcal{E}^{(0)}$ be  a local noise channel satisfying the condition (I) and (II) in the Theorem \ref{thm_1}.
    Then, \red{the average of the number of copies of $\rho_{\mathrm{noisy}}$} required to perform unbiased estimation of $\ev*{\hat{X}}$ with standard deviation $\varepsilon$ over $\{U_1, ..., U_L\}$ in the above setup is lower bounded as
    \begin{eqnarray}
    \label{eq_thm_4}
        \mathbb{E}_{U_1\cdots U_L}[N]
    &\geq& O\qty(\qty(\frac{4^n\nu((\mathcal{E}^{(0)})^{-1})^n-2^n\eta((\mathcal{E}^{(0)})^{-1})^n}{4^n-1})^L).
    \end{eqnarray}
\end{thm}

\begin{proof}
    Since we assume condition (II) for noise $\mathcal{E}$, 
    there exists $\beta>0$ such that $\mathcal{E}(\rho) - \beta \hat{I}/d>0$ for all $n$-qubit state $\hat{\rho}$.
    Thus, in the same way as in the proof of Theorem \ref{thm_3}, we can lower bound the average of the cost $N$ as 
    \begin{eqnarray}
        \mathbb{E}_{U_1\cdots U_L}[N] 
        \geq  \frac{1}{\varepsilon^2}\frac{d}{d^2-1}\mathrm{tr}[\hat{X}^2] \qty(\beta\mathbb{E}_{U_2\cdots U_{L-1}}[\nu(\mathcal{N}^{-1}_{L-1})] - \frac{\beta}{d^2} - \frac{d-1}{d^2}).
    \end{eqnarray}

    Next, let us analyze $\mathbb{E}_{U_l}[\nu(\mathcal{N}^{-1}_{l})]$ and $\mathbb{E}_{U_l}[\eta(\mathcal{N}^{-1}_{l})]$.
    In the same way as in the proof of Theorem \ref{thm_3}, we obtain
    \begin{eqnarray}
        \mathbb{E}_{U_l}[\nu(\mathcal{N}^{-1}_{l})] 
        &=& \frac{d^2\nu(\mathcal{E}^{-1})-1}{d^2-1}\nu(\mathcal{N}^{-1}_{l-1}) + \frac{d-d\nu(\mathcal{E}^{-1})}{d^2-1}\eta(\mathcal{N}^{-1}_{l-1}),\\
        \mathbb{E}_{U_l}[\eta(\mathcal{N}^{-1}_{l})] 
        &=& \frac{d^2\eta(\mathcal{E}^{-1})-d}{d^2-1}\nu(\mathcal{N}^{-1}_{l-1}) + \frac{d^2-d\eta(\mathcal{E}^{-1})}{d^2-1}\eta(\mathcal{N}^{-1}_{l-1}).
    \end{eqnarray}
    Since the eigenvalues of the matrix 
    \begin{eqnarray}
        \frac{1}{d^2-1}
        \begin{pmatrix}
            d^2\nu(\mathcal{E}^{-1})-1 & d-d\nu(\mathcal{E}^{-1}) \\
            d^2\eta(\mathcal{E}^{-1})-d & d^2-d\eta(\mathcal{E}^{-1}) \\
        \end{pmatrix}
    \end{eqnarray}
    are $1$ and $(d^2\nu(\mathcal{E}^{-1})-d\nu(\mathcal{E}^{-1}))/(d^2-1)$, we obtain
    \begin{eqnarray}
        \mathbb{E}_{U_2\cdots U_{L-1}}[\nu(\mathcal{N}^{-1}_{L-1})] = O\qty(\qty(\frac{d^2\nu(\mathcal{E}^{-1})-d\eta(\mathcal{E}^{-1})}{d^2-1})^L).
    \end{eqnarray}
    By using the multiplicativity of $\nu$ and $\eta$ under local noise, we finally obtain
    \begin{eqnarray}
        \mathbb{E}_{U_1\cdots U_L}[N]
    &\geq& O\qty(\qty(\frac{4^n\nu((\mathcal{E}^{(0)})^{-1})^n-2^n\eta((\mathcal{E}^{(0)})^{-1})^n}{4^n-1})^L).
    \end{eqnarray}
\end{proof}

\red{
We note that Theorem \ref{thm_2} (Theorem \ref{thm_3} and Theorem \ref{thm_4}) only state the result for QEM methods where we only utilize the copies of noisy layered circuits without any modification, such as virtual distillation or rescaling method we have presented here.
However, we believe that we can expand this result for more general QEM methods that can be described by the virtual circuit.
}

\section{Application of Theorem 2 to specific noise models}
In this section, we apply Theorem \ref{thm_2} (Theorem \ref{thm_3} and Theorem \ref{thm_4}) to specific noise models, namely, the local depolarizing noise, local dephasing noise, and the amplitude damping noise.

First, we consider the case where all unitary gates are followed by the local depolarizing noise $\mathcal{E}_l = \bigotimes_{i=1}^n \mathcal{E}_l^{(i)}$ with $\mathcal{E}_l^{(i)}(\hat{\rho}) = (1-\frac{3}{4}p)\hat{\rho} + \frac{1}{4}p(X\hat{\rho} X + Y\hat{\rho} Y + Z\hat{\rho} Z)$.
In terms of Bloch vector, $\mathcal{E}_l^{(i)}$ maps Bloch vector $\vb*{\theta}$ to $(1-p)\vb*{\theta}$.
Since the inverse of $\mathcal{E}_l^{(i)}$ can be written as $(\mathcal{E}_l^{(i)})^{-1}(\hat{\rho}) = \frac{4-p}{4-4p}\hat{\rho} - \frac{p}{4-4p}(X\hat{\rho} X + Y\hat{\rho} Y + Z\hat{\rho} Z)$, we obtain
\begin{eqnarray}
    \nu((\mathcal{E}_l^{(i)})^{-1}) = \qty(\frac{4-p}{4-4p})^2 + 3 \qty(\frac{p}{4-4p})^2 = 1 + \frac{3}{2}p + O(p^2).
\end{eqnarray}
Thus, the average cost is lower bounded as
\begin{eqnarray}
    \mathbb{E}_{U_1\cdots U_L}[N] \geq O\qty(\qty(1+\frac{3}{2}\frac{4^n}{4^n-1}p + O(p^2))^{nL}).
\end{eqnarray}
We have compared this bound in Fig.~\ref{fig_numerics}(b) in the main text with some QEM methods.

Next, we discuss the quantum circuit with local dephasing noise $\mathcal{E}_l = \bigotimes_{i=1}^n \mathcal{E}_l^{(i)}$ with $\mathcal{E}_l^{(i)}(\hat{\rho}) = (1-\frac{1}{2}p)\hat{\rho} + \frac{1}{2}pZ\hat{\rho} Z$.
In terms of Bloch vector, $\mathcal{E}_l^{(i)}$ maps Bloch vector $(\theta_x,\theta_y,\theta_z)$ to $((1-p)\theta_x,(1-p)\theta_y,\theta_z)$.
Since the inverse of $\mathcal{E}_l^{(i)}$ can be written as $(\mathcal{E}_l^{(i)})^{-1}(\hat{\rho}) = \frac{2-p}{2-2p}\hat{\rho} - \frac{p}{2-2p}Z\hat{\rho} Z$, we obtain
\begin{eqnarray}
    \nu((\mathcal{E}_l^{(i)})^{-1}) = 1 + p + O(p^2).
\end{eqnarray}
Thus, the average cost is lower bounded as
\begin{eqnarray}
    \mathbb{E}_{U_1\cdots U_L}[N] \geq O\qty(\qty(1+\frac{4^n}{4^n-1}p + O(p^2))^{nL}).
\end{eqnarray}
We compare this bound in Fig. \ref{fig_numerics_SM}  with some QEM methods.

Finally, let us think of the amplitude damping noise $\mathcal{E}_l = \bigotimes_{i=1}^n \mathcal{E}_l^{(i)}$ with $\mathcal{E}_l^{(i)}(\hat{\rho}) = E_1\hat{\rho} E_1^\dag + E_2\hat{\rho}E_2^\dag$, where
\begin{eqnarray}
    E_1 = \begin{pmatrix}
            1 & 0 \\
            0 & \sqrt{1-p} \\
          \end{pmatrix},\;\;
    E_2 = \begin{pmatrix}
            0 & \sqrt{p} \\
            0 & 0 \\
          \end{pmatrix}.
\end{eqnarray}
In terms of Bloch vector, $\mathcal{E}_l^{(i)}$ maps Bloch vector $(\theta_x,\theta_y,\theta_z)$ to $(\sqrt{1-p}\theta_x,\sqrt{1-p}\theta_y,(1-p)\theta_z+p)$.
Strictly speaking, this noise does not satisfy condition (II) in the main text, so we cannot apply Theorem \ref{thm_4}.
However, the state after the noise is applied is almost always full rank, so we can expect that the scaling in Theorem \ref{thm_4} also holds in this case.
We note that this assumption can also be justified \red{ when $\rho_{\mathrm noisy}$ is full rank} by the numerical result obtained in the next section.

Since the inverse of $\mathcal{E}_l^{(i)}$ can be written as $(\mathcal{E}_l^{(i)})^{-1}(\hat{\rho}) = E_1'\hat{\rho}(E_1')^\dag - E_2'\hat{\rho}(E_2')^\dag$ where
\begin{eqnarray}
    E_1' = \begin{pmatrix}
            1 & 0 \\
            0 & \frac{1}{\sqrt{1-p}} \\
          \end{pmatrix},\;\;
    E_2' = \begin{pmatrix}
            0 & \frac{\sqrt{p}}{1-p} \\
            0 & 0 \\
          \end{pmatrix},
\end{eqnarray}
we obtain
\begin{eqnarray}
    \nu((\mathcal{E}_l^{(i)})^{-1}) &=& 1 + p + O(p^2),\\
    \eta((\mathcal{E}_l^{(i)})^{-1}) &=& \frac{1}{2} + O(p^2).
\end{eqnarray}
Thus, we can expect that the average cost is lower bounded as
\begin{eqnarray}
    \mathbb{E}_{U_1\cdots U_L}[N] \geq O\qty(\qty(1+\frac{4^n}{4^n-1}p + O(p^2))^{nL}).
\end{eqnarray}
We have compared this bound in Fig.~\ref{fig_numerics}(c) in the main text with some QEM methods.

\begin{figure}[ht]
    \begin{center}
        \includegraphics[width=70mm]{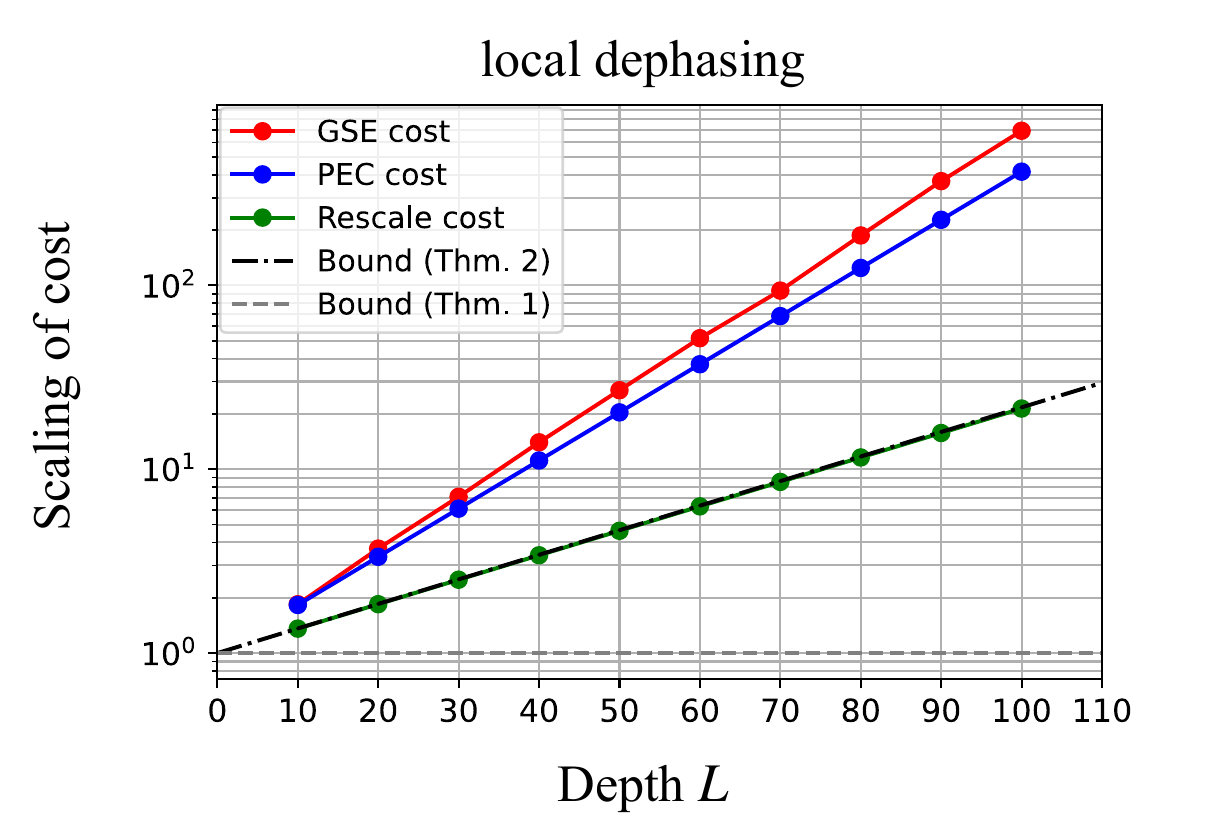}
        \caption{Scaling of the cost to perform QEM methods for random Clifford circuit of $n=3$ qubits under local dephasing noise with error rate $p=0.01$.
        The red, blue, and green lines denote the sampling overhead of generalized subspace expansion~\cite{yoshioka2022generalized, yang2023dual} using power subspace, the probabilistic error cancellation as derived in Ref.~\cite{takagi2020optimal}, and the rescaling technique as explained in the main text. The rescaling factor is $(1-p)^{-2nL4^{n-1}/(4^n-1)}$.
        Bound (Thm. 1) and Bound (Thm. 2) represent the lower bound of the cost obtained from Theorem 1 and Theorem 2, respectively.
        Note that GSE and the rescaling methods do not completely eliminate the errors, while we confirm significant reduction of bias.
        }
        \label{fig_numerics_SM}
    \end{center}
\end{figure}

\section{Convergence of noise to global depolarizing noise} \label{sec:circuit_local_dep}
In this section, we explain the details of our numerical simulation implying that the noise in the deep depth limit can be characterized by the global depolarizing noise.
As we have stated in the main text, we calculated the singular values of the unitial part of the Pauli transfer matrix of the effective noise channels $\mathcal{E}_m'$, \red{whose inverse upper bounds the quantum fisher information matrix of the circuit as stated in Lemma \ref{lem_1}}.
We have observed that, in various kinds of random circuits, the effect of each noise becomes indiscriminable from that of the global depolarizing noise for sufficiently large circuit depth $L$, and its noise strength can be given by the geometric mean of the singular values of the unitial part of the Pauli transfer matrix of each noise channel $\mathcal{E}_{lm}$.
This can be interpreted as the effects of noise being mixed and averaged out by the sequence of the noisy random gate.
In Fig.~\ref{fig_numerics_SM2} we have compared several different random circuit structures to find the generality of the phenomenon.
To be concrete, we have compared the hardware-efficient (HE) random circuit structure where 2-qubit gates (controlled-Z) can only be performed on two adjacent qubits under linear connectivity, a random circuit with each layer consisting of 2-qubit random unitary operators that act on randomly chosen pairs of sites (2Q random), and finally the Haar random unitary (Haar).
Surprisingly, the absolute error from the ideal value of $k$ determined from the geometric mean of $A_{l}$ is suppressed as $O(1/\sqrt{L})$ for either choice of the random circuit structure,
which implies the generality of this phenomenon.

We can obtain several important consequences from these results.
First, these results indicate the exponential decay of the quantum Fisher information matrix with respect to qubit count, meaning that the cost of QEM grows exponentially with qubit count.
Theorem \ref{thm_2} (Theorem \ref{thm_3},\ref{thm_4}) only states the results for the unitary 2-design, but our numerical results allow us to generalize this result even more.
Second, these results imply that just rescaling the measurement results may allow us to mitigate the effect of noise as in the case of global depolarizing noise.
Although we cannot obtain an unbiased estimator, we confirmed a reduction of the bias by orders of magnitude \red{through this method}.
What is more, the scaling of the cost of this rescaling method matches our average bound in the random circuit as shown in Fig. \ref{fig_numerics_SM} as well as Fig.~\ref{fig_numerics} in the main text.

\begin{figure}[ht]
    \begin{center}
        \includegraphics[width=0.95\linewidth]{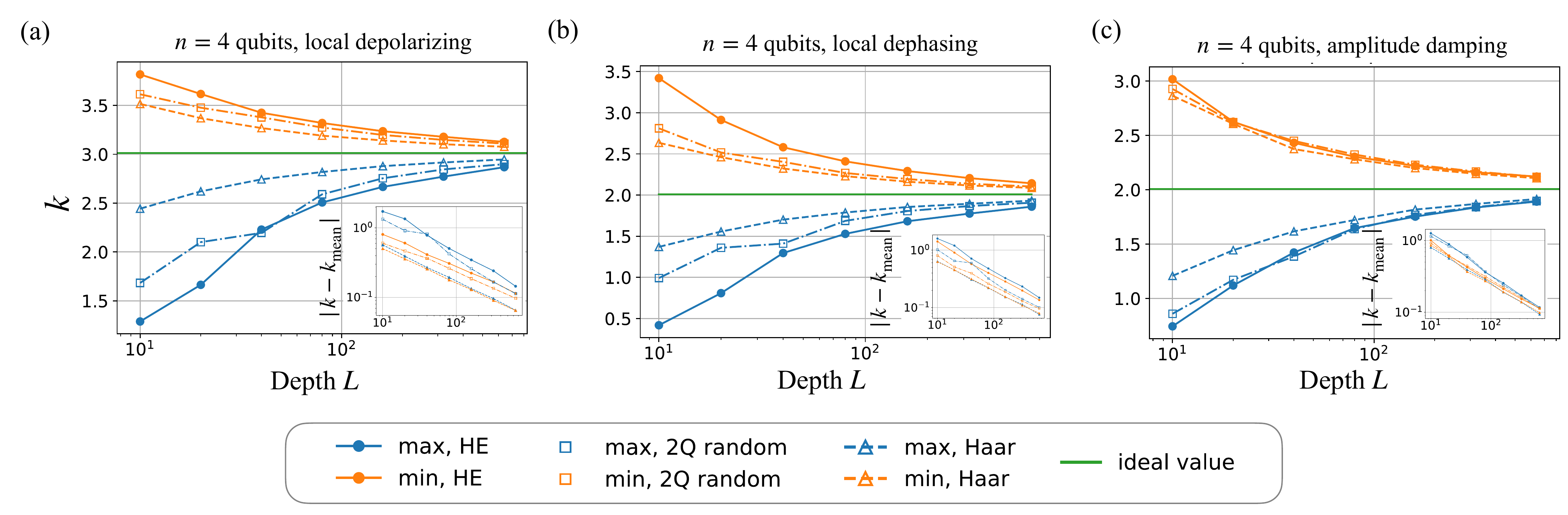}
        \caption{Convergence of (a) local depolarizing, (b) local dephasing, and (c) amplitude damping noise into global depolarizing noise under random circuits of $n=4$ qubits with error rate $p=0.0001$.
        We denote by $(1-p)^{kL}$ the singular values of the unital part of the Pauli transfer matrix for the effective noise channel at each depth $L$, where $k$ for the maximal and minimal ones are plotted in this figure. As is highlighted in the inset, we find that all $k$'s approach the geometric mean $k_{\rm mean}$ of the singular values \red{for each noise channel} with its fluctuation scaling as $O(1/\sqrt{L})$, implying the convergence to the global depolarizing noise.
        For instance, $k_{\rm mean} = 3n4^{n-1}/(4^n-1)$ for local depolarizing and $k_{\rm mean} = 2n4^{n-1}/(4^n-1)$ for amplitude damping.
        Here, we consider three class of random circuits: hardware-efficient ansatz with random parameters, 2-qubit random unitary between random pairs, and Haar random unitary}
        \label{fig_numerics_SM2}
    \end{center}
\end{figure}

\section{Notes on QEM method using some knowledge of noiseless states}
In this section, we discuss some caveats in applying our results to QEM methods which use some knowledge of noiseless states, such as subspace expansion \redd{or symmetry expansion}.
For the numerical simulation performed to obtain the cost of generalized subspace expansion~\cite{yoshioka2022generalized},  we optimize parameter $c_1$ and $c_2$ so that the error mitigated state $(c_1I/2^n + c_2\mathcal{E}'(\rho))(c_1I/2^n + c_2\mathcal{E}'(\rho))^\dagger = \abs{c_1}^2/4^nI + \mathrm{Re}(c_1c_2^*)/2^{n-1}\mathcal{E}'(\rho) + \abs{c_2}^2\mathcal{E}'(\rho)^{\redd{2}}$ reaches the noiseless state $\rho$, where $\mathcal{E}'$ represents the effective noise channel, and define $\qty(\abs{\mathrm{Re}(c_1c_2^*)}/2^{n-1} + \abs{c_2}^2)^2$ as the cost of performing generalized subspace expansion.
\redd{Here, we used that we can obtain the expectation value corresponding to $\mathcal{E}'(\rho)^{2}$ only from measurements of single copies of $\mathcal{E}'(\rho)^{2}$ by using dual-state purification~\cite{yang2023dual}.}
In the case of global depolarizing noise, the optimal values are $c_1 = -\frac{1-(1-p)^L}{(1-p)^L}$ and $c_2 = \frac{1}{(1-p)^L}$, which is independent of the noiseless state.

However, for other noise models such as local depolarizing noise or dephasing noise, optimal values depend on the noiseless state, which means that the POVM measurement $\mathcal{M}$ should also depend on the noiseless state.
This situation gets even worse when we include the optimizing process in the POVM measurement $\mathcal{M}$.
Therefore, generalized subspace expansion is actually outside the realm of our result, because $\mathcal{M}$ should not depend on the noiseless state when we apply quantum Cram\'{e}r-Rao inequality.
This is also the case for the usual subspace expansion.
We observe that the sampling cost after the optimization of parameters actually breaks our lower bound when the subspace is spanned to the original space.
Therefore, we must take care of applying our results to error mitigation methods such as subspace expansion, where the estimation method depends on the noiseless state we want to prepare.

\section{Comparison with related studies}
\red{
In this section, we briefly explain the differences between our work and some related studies.
Let us first introduce some related previous research.
Refs.~\cite{aharonov1996limitations} showed that the relative entropy $D$ between the output of the noisy circuit exposed to local depolarizing noise (which we have denoted as $\mathcal{E}'_m(\hat{\rho}(\vb*{\theta}))$) and the maximally mixed state $\hat{I}/2^n$ converges exponentially to 0 with respect to circuit depth as $D(\mathcal{E}'_m(\hat{\rho}(\vb*{\theta}))||\hat{I}/2^n)\leq n(1-p)^{L}$.
Thus, they concluded that any noisy quantum circuit without quantum error correction must have a size (= number of qubits) exponential in its depth in order to obtain meaningful samples.
In other words, we can only obtain meaningful samples from circuits whose depth scales as $L=O(\log(n))$.
Subsequently, Refs.~\cite{kastoryano2013quantum} proved that the convergence of relative entropy becomes quadratically faster as  $D(\mathcal{E}'_m(\hat{\rho}(\vb*{\theta}))||\hat{I}/2^n)\leq n(1-p)^{2L}$.
Refs.~\cite{takagi2022fundamental} applied these results to the analysis of QEM: they showed the exponential growth in the maximum estimator spread, or the range of possible values for the estimator.
Refs.~\cite{wang2021can} also derived similar results.
However, bounds on the maximum estimator spread only gave a sufficient condition for the cost of QEM, and thus there was no known lower bound on the sample complexity for unbiased QEM to our knowledge.
}

\red{
Along with our results, Refs.~\cite{takagi2022universal} and Refs.~\cite{quek2022exponentially} also characterized the resources needed for QEM.
By using the convergence of the relative entropy, they showed for generic noisy layered circuits that the sampling cost of QEM grows exponentially with the circuit depth.
Even though these results give the same scaling as our results obtained in Theorem \ref{thm_1}, our result provides a tighter and achievable bound,  which allows us to derive the optimal QEM method.
Refs.~\cite{quek2022exponentially} also obtained similar results to Theorem \ref{thm_2}.
However, our average bound for non-unital noise Eq. (\ref{eq_thm_4}) is quadratically tighter than the bound stated in Eq. (187) of  Refs.~\cite{quek2022exponentially} because of the factor $1/2$ in the exponent.
}

\end{document}